\documentclass[11pt,a4paper,reqno,french,tikz]{amsart}
\usepackage[utf8]{inputenc} \usepackage[T1]{fontenc} 
\usepackage{amsthm,amsmath,amsfonts,amssymb,amsxtra,appendix,bookmark,dsfont,bm,mathrsfs,amstext,amsopn,mathrsfs,mathtools,comment,cite,hyperref,color}
\usepackage{graphicx} 
\usepackage{enumitem}
\usepackage{makecell} 
\usepackage{tikz,tikz-cd}
\usetikzlibrary{datavisualization}
\usetikzlibrary{datavisualization.formats.functions}
\usetikzlibrary{shapes.callouts}
\tikzset{
  level/.style   = { ultra thick, blue },
  connect/.style = { dashed, red },
  notice/.style  = { draw, rectangle callout, callout relative pointer={#1} },
  label/.style   = { text width=2cm }
}
\usepackage{pgfplots}
\usepackage{bbm} 
\usepackage{tabu}
\usepackage{xfrac} 
\usepackage{nicefrac} 
\usepackage{stmaryrd} 
\usepackage{macros_art}
\usepackage{macros}
\newcommand{\rop}{\wt{\ro}} 

\newcommand{\gam}{\cP} 

\newcommand{\pt}{\mathscr{P}} 
\newcommand{\gamp}{\pa{1-\gam(v)}} 
\newcommand{\spa}{\cV} 
\newcommand{\spp}{L^p+L^{\ii}} 
\newcommand{\sppe}{L^p+L^{\ii}_{\ep}} 
\newcommand{\spn}[1]{\cV_N^{(#1)}} 
\newcommand{\sps}[2]{\cV_{#2}^{(#1)}} 
\newcommand{\spf}[1]{\cV_{N,\partial}^{(#1)}} 
\newcommand{\spfn}[2]{\cV_{#2,\partial}^{(#1)}}
\newcommand{\spnm}[1]{\cV^{(#1)}_{N,\text{meta}}} 
\newcommand{\deli}{\pa{-\Delta_i+1}} 
\newcommand{\del}{\pa{-\Delta+1}} 

\newcommand{\chh}{\sch_{k,\ii}} 
\newcommand{\delk}{(-\Delta+1)^{\frac{k}{2}}} 

\newcommand{\exc}[1]{E^{(#1)}_N}
\newcommand{\excn}[2]{E^{(#1)}_{#2}}
\newcommand{\hn}{H_N}

\newcommand{\sign}{\Sigma_N}
\newcommand{\signn}[1]{\Sigma_{#1}}

\usepackage{ifthen}\newboolean{versionlongue}\setboolean{versionlongue}{false} 
\usepackage{scalerel,stackengine,wasysym} 

\usepackage{array}
\newcommand{\PreserveBackslash}[1]{\let\temp=\\#1\let\\=\temp}
\newcolumntype{C}[1]{>{\PreserveBackslash\centering}p{#1}}
\newcolumntype{R}[1]{>{\PreserveBackslash\raggedleft}p{#1}}
\newcolumntype{L}[1]{>{\PreserveBackslash\raggedright}p{#1}}

\title[]{Some properties of the\\potential-to-ground state map\\in quantum mechanics}
\author[L. Garrigue]{Louis Garrigue}
\address{CERMICS, \'Ecole des ponts ParisTech, 6 and 8 av. Pascal, 77455 Marne-la-Vallée, France} 
\email{louis.garrigue@enpc.fr}
\date{\today}

\begin{document} 
\begin{abstract} 
We analyze the map from potentials to the ground state in static many-body quantum mechanics. We first prove that the space of binding potentials is path-connected. Then we show that the map is locally weak-strong continuous and that its differential is compact. In particular, this implies the ill-posedness of the Kohn-Sham inverse problem.
\end{abstract}
\maketitle
\flushbottom
\setcounter{tocdepth}{1} 

The potential-to-eigenstate map is one of the main objects in quantum mechanics, since its knowledge enables to deduce many physical quantities. 
 The mathematical structure of this map is very rich, it relates to degenerate perturbation theory and Rayleigh-Schrödinger series \cite{ReeSim4}, adiabaticity \cite{Teufel03}, the topology of binding potentials, and so on. Moreover, in Density Functional Theory, the space of potential-representable densities is important to know in order to characterize Kohn-Sham potentials \cite{Lieb83b}, and the potential-to-eigenstate map contains this information.

In this work, we prove some mathematical properties of this map. The natural starting space is the set of potentials which are able to bind $N$ particles, and it has no known simple characterization. We show that it is path-connected when degeneracies are allowed, implying that the set of potential-representable densities is also path-connected. Then we prove that the potential-to-ground states map is locally weak-strong continuous, and that restricted to potentials having a non-degenerate ground state, it is smooth and has a compact differential. Next, we show that the potential-to-ground state energy map is singular on degenerate potentials, in absence of interactions. These results allow us to deduce that the Kohn-Sham problem of Density Functional Theory is ill-posed on a bounded set $\Omega$, when restricted to the non-degenerate case. 

When our proofs allow it, we state the results in the case of excited states. We remark that the ground state map is special in the sense that most of its properties do not simply extend to excited states.

\subsection*{Acknowledgement}
I warmly thank Mathieu Lewin, my PhD director, for having advised me during this work. This project has received funding from the European Research Council (ERC) under the European Union's Horizon 2020 research and innovation programme (grant agreement MDFT No 725528).

\section{Main results: properties of the map}\label{pco} 
\subsection{Definitions}

We consider an open connected set $\Omega \subset \R^d$ with Lipschitz boundary, in which the particles live. We consider external potentials $v\in (L^p+L^{\ii})(\Omega,\R)$ and an even positive interaction potential $w \in (L^p+L^{\ii})(\R^d,\R)$. The exponent can be $p > \max \pa{\frac{2d}{3},2}$ in case we need to apply unique continuation for the many-body Schrödinger operator \cite{Garrigue19}, or it can be 
\begin{align}\label{expot}
	p = 1 \tx{ for } d = 1, \bhs p > 1 \tx{ for } d = 2, \bhs p = \f{d}{2} \tx{ for } d \ge 3,
\end{align}
otherwise. Our space includes Coulomb-like singularities in $d=3$ involved in the physical situation. We consider the $N$-particle Schr\"odinger operator
\begin{equation}
\hn(v) \df \sum_{i=1}^N -\Delta_i  + \sum_{1\leq i < j \leq N} w(x_i-x_j) + \sum_{i=1}^N v(x_i),
\end{equation}
acting on the space of antisymmetric spinless wavefunctions $L^2\ind{a}\pa{\Omega^N} \df \wedge^N L^2(\Omega)$, and where $-\Delta$ is the Dirichlet Laplacian. We consider its form domain $(H^1_0 \cap L^2\ind{a})(\Omega^N)$ of antisymmetric functions vanishing on the boundaries, and also denote by $H_N(v)$ the associated Friedrichs operator. The results presented in this document would not depend on the boundary conditions. We denote by $\cE_v(\p) \df \ps{\p,\hn(v) \p}$ the energy functional, and by
\begin{align*}
\ro_{\p}(x) \df N \int_{\R^{d(N-1)}} \ab{\p}^2(x,x_2,\dots,x_N) \d x_2 \cdots \d x_N
\end{align*}
the one-body density of a state. We recall \cite[Section 12.1]{LieLos01} that the $k^{\tx{th}}$ excited energies are
 \begin{align}\label{mlm}
	 \exc{k}(v) \df \mysup{A \subset H^1\ind{a}(\Omega^N) \\ \dim A = k} \myinf{\p \in A^{\perp} \\ \int \ab{\p}^2 = 1} \cE_v(\p) = \myinf{A \subset H^1\ind{a}(\Omega^N) \\ \dim A = k+1} \mymax{\p \in A \\ \int \ab{\p}^2 = 1} \cE_v(\p),
 \end{align}
 where $A$ are linear subsets, the ground energy is thus $\exc{0}(v)$. We also denote by $\sign(v) \df \inf \sigma\ind{ess}\bpa{\hn(v)}$ the bottom of the essential spectrum of $\hn(v)$. 
 We define $\spa = \spp$ or 
\footnote{Let us recall \cite{ReeSim4} that 
 \begin{align*}
	 & L^p+L^{\ii}_{\ep}\df \acs{ f \in (L^p\hspace{-0.1cm}+L^{\ii})(\R^d,\R) \st \forall \ep > 0, \exists g_{\ep}, h_{\ep}, \hs f\hspace{-0.1cm} = g_{\ep}\hspace{-0.1cm} + h_{\ep}, \nor{h_{\ep}}{L^{\ii}} \hspace{-0.1cm} \le \ep, g_{\ep}\hspace{-0.1cm} \in\hspace{-0.1cm} L^p}.
 \end{align*}} 
 $\spa = \sppe$ depending on the situation. When we work under the condition that $0 \le w \in \sppe$, the HVZ theorem \cite{Zhislin60b,VanWinter64,Hun66,Lewin11} says that $\excn{0}{N-1}(v) = \sign(v)$ whenever $v \in \sppe$. 
 We now introduce the space of non-degenerate binding potentials for ground and excited states
\begin{align}\label{ddd}
	\spn{k}  \df \acs{v \in \spa \bigst \exc{k}(v) \sle \sign(v), \hs \dim \Ker \bpa{ \hn(v) - \exc{k}(v)} = 1},
\end{align}
the space of (possibly degenerate) binding potentials
\begin{align*}
	\spf{k} \df \acs{v \in \spa \bigst \exc{k}(v) \sle \sign(v)}, 
\end{align*}
and the most general set of metastable binding potentials \footnote{Elements $v$ of $\spnm{0} \backslash \spf{0}$ satisfy ${\exc{0}(v) = \sign(v)}$, here is an example. Take $N=1$, $d \ge 5$, $\p(x) = c(1+x^2)^{1-\f{d}{2}}$ where $c$ normalizes $\p$. We have $-\Delta \p + v \p = 0$ with $v(x) = - d(d-2) (1+x^2)^{-2} \in L^{d/2} \cap L^{\ii}$. We know that $\p$ is the ground state since it is strictly positive everywhere. Hence $\exc{0}(v) = \sign(v) = 0$.}
 \begin{align*}
	 \spnm{k} \df \acs{ v \in \spa \bigst \dim \Ker \bpa{ \hn(v) - \exc{k}(v)} \ge 1}.
 \end{align*}
Those are all locally endowed with the norm of $\spp$, recalled below in \eqref{ddef}. 
 They satisfy $\spn{k} \subset \spf{k} \subset  \spnm{k} \subset \spp$, and we will not work on $\spnm{k}$ in this document.

\subsection{Path-connectedness of the space of binding potentials}

By perturbation theory \cite{Kato,ReeSim4}, $\spn{k}$ and $\spf{k}$ are open in $\spa \df \spp$. The manifold structure is then canonical and locally flat. The injections $\spn{k}\xhookrightarrow{} \spa$ and $\spf{k} \xhookrightarrow{}\spa$ make them smooth embedded manifolds of $\spa$, and $\tx{T}_v \spn{k} = \spa$. We now show that the set of trapping electric potentials is path-connected.

\begin{theorem}[Path-connectedness of the space of binding potentials]\label{bindos} Take $\spa = (\spp)(\R^d,\R)$, $p$ as in \eqref{expot}, and take $w \in \sppe$ with $w \ge 0$. Then $\cap_{n=1}^N \spfn{0}{n}$ is path-connected.
\end{theorem}

\begin{remark}
	For instance we can connect all the elements to a well $-c_{d,N} \indic_{B_1}$, where the constant $c_{d,N} > 0$ is chosen large enough so that $-c_{d,N}\indic_{B_1}$ belongs to $\cap_{n=1}^N \spfn{0}{n}$. We naturally conjecture that $\spfn{0}{N+1} \subset \spfn{0}{N}$ 
	for any $N  \ge 1$ and any interaction $w \ge 0$, which would imply $\cap_{n=1}^N \spfn{0}{n} = \spf{0}$.
\end{remark}
It would be interesting to know whether the same result holds for $\spn{0}$. We present other remarks in Section \ref{secpot}, where we provide the proof of Theorem \ref{bindos}. Also, we will explain in the proofs that to any path connecting two given binding potentials, there is a corresponding piecewise real analytic path $t \mapsto \p\ex{0}(t)$ of ground states connecting two of the initial ground states. Hence there also exists a corresponding path of densities $t \mapsto \ro_{\p\ex{0}(t)}$.
\begin{corollary}\label{rkb}
The set of $v$-representable densities
 \begin{align*}
	 \acs{ \ro_{\p} \in L^1(\R^d,\R_+) \bigst \p \tx{ is the ground state of $\hn(v)$ for some } v \in \cap_{n=1}^N \spfn{0}{n}}
 \end{align*}
is path-connected.
\end{corollary}

\subsection{Local weak-strong continuity of $v \mapsto \p\ex{0}(v)$}\label{chprop} 
Let us define
 \begin{align*}
	 \Sbb \df \acs{ \p \in L\ind{a}^2\bpa{\Omega^N,\C} \st \nord{\p} = 1}, \bhs  H^k\ind{p} \df \frac{H^k(\Omega^N) \cap \Sbb}{ S^1},
 \end{align*}
 where $S^1$ represents the circle of phase factors. The space of rays $H^k\ind{p}$ identifies two vectors equal up to a global phase. We define the maps
\begin{align*}
	\p\ex{k} : 
\begin{array}{ccc}
\spn{k} & \lra & H^1\ind{p}(\Omega) \\
	v & \longmapsto & \p\ex{k}(v), \\
\end{array}
\end{align*}
where $\p\ex{k}(v)$ denotes the non-degenerate $k\expo{th}$ excited eigenstate of $H_N(v)$, being the ground state when $k=0$. We define
 \begin{align}\label{defpm}
	  \bpa{\hn(v)-\exc{k}(v)}\iv_{\perp}
 \end{align}
to be $0$ on $\C \p\ex{k}(v)$, and on $\acs{\p\ex{k}(v)}^{\perp}$ to be the inverse of the restriction of $\hn(v)-\exc{k}(v)$ to $\acs{\p\ex{k}(v)}^{\perp}$. We define an equivalence relation $\sim$ on sets of potentials, by $v \sim u$ if $v-u$ is constant. The following theorem is the main one of our work.

\begin{theorem}[Regularity and local weak-strong continuity]\label{proposs} We take $p$ as in \eqref{expot}, $\spa = (\spp)(\Omega,\R)$ and $w \in (\spp)(\R^d,\R)$. \smallskip

$(i - \tx{Smoothness}).$ The map $\p\ex{k}$ is $\cC^{\ii}$ from $\spn{k}$ to $H^1\ind{p}$. When $\spa = (\spp)/\sim$ and $p > \max(2d/3,2)$, the map $\p\ex{0}$ is injective. \smallskip

	$(ii - \tx{Compactness of the differential}).$ For any $v \in\spn{k}$, the differential $\d_v \p\ex{k} : (\spp)(\Omega,\R) \ra \acs{\p\ex{k}(v)}^{\perp} \cap H^1$ equals
\begin{align}\label{peq}
	\bpa{\d_v \p\ex{k}} u = -\bpa{\hn(v)-\exc{k}(v)}\iv_{\perp} \bpa{\Sigma_{i=1}^N u(x_i)} \p\ex{k}(v),
\end{align}
	where we considered the definition \eqref{defpm}. Moreover, for all $v \in \spn{k}$ $\d_v \p\ex{k}$ is compact, and
\begin{align*}
	\nor{\bpa{\d_v \p\ex{k}}u}{H^1}^2 \le c_{v} \nor{u}{\spp} \int_{\Omega} \ab{u} \ro_{\p\ex{k}(v)}.
 \end{align*}
	If $p > \max(2d/3,2)$, $\pa{\d_v \p\ex{0}} u = 0$ implies that $u$ is constant.

	$(iii - \tx{Local weak-strong continuity}).$ Let $p$ be as in \eqref{expot}, with $p > d/2$ when $d \ge 3$, $w \in \sppe$, $w \ge 0$ and $\spa = (\sppe)(\Omega,\R)$ in the definition \eqref{ddd} of $\spn{0}$. Let $\Lambda \subset \Omega$ be a bounded open subset of $\Omega$. Assume that $v,v_n \in \spf{0}$ with
 \begin{align*}
	  	v_n \wra v, \bhs\bhs v_n \indic_{\Omega \backslash \Lambda} \ra v\indic_{\Omega \backslash \Lambda},
 \end{align*}
resp. weakly and strongly in $(L^p+L^{\ii})(\Omega,\R)$. Then $\exc{0}(v_n) \ra \exc{0}(v)$ and for $n$ large enough $v_n \in \spf{0}$. Moreover, for any sequence $\p_n$ of approximate minimizers, that is satisfying $\cE_{v_n}(\p_n) \le \exc{0}(v_n) + \ep_n$ where $0 \le \ep_n \ra 0$ and $\nor{\p_n}{L^2} = 1$, then 
 \begin{align*}
	 P_{\Ker \bpa{\hn(v)-\exc{0}(v)}^{\perp}} \p_n \ra 0
 \end{align*}
strongly in $H^1(\Omega^N)$.  \smallskip

	$(iv - \tx{Compactness for } \Omega \tx{ bounded}).$ Let $p$ be as in \eqref{expot}, with $p > d/2$ when $d \ge 3$, and $\spa = (\spp)/\sim$. If $\Omega$ is bounded, $v \mapsto \p\ex{0}(v)$ is compact and $(\p\ex{0})\iv$ is discontinuous.
\end{theorem} 
In \cite{Lampart20}, Lampart proved a weak-strong continuity result in the dynamical case. In $iii)$ above and in all this document, $P_V$ denotes the orthogonal projection onto the vector subspace $V$, and $P^{\perp}_V \df 1- P_V$. In particular when $v_n \wra v$ weakly in $L^p$ with $v_n, v \in \spn{0}$ and under the above assumptions, then $\p(v_n) \ra \p(v)$ strongly in $H^1\ind{p}$. Such input-output maps involving second order differential equations are generically locally compact \cite{HasAleRom17}. In particular, Theorem \ref{proposs} $(iii)$ implies that quantum particles are insensitive to highly oscillating local electric fields.

If $\Omega = \R^d$, $\p\ex{0}$ is not weak-strong continuous because of simple counterexamples. For instance by taking $u,v \in \spn{0}$ with $\exc{0}(u) \sle \exc{0}(v)$, and $v_n(x) \df u(x-n)$, then $\p\ex{0}(v_n + v) \wra 0 \neq \p\ex{0}(v)$. However, up to translations and for $w \ge 0$, it would be possible to state a weak-strong continuity result when $d \ge 3$. 
To prove it, one could use concentration-compactness principles, see for instance \cite{Lions84,Lions84b,Lions85a,Lions85b,Lieb83,Struwe,Lewin10,Lewin11}. Assuming that $\nor{v_n}{L^p}$ is bounded, this would consist in extracting $K \in \N$ \apo{bubbles} $v_1,\dots,v_K$ from the sequence $v_n$, with $K$ large enough so that 
\begin{align*}
\sup \acs{ \nor{v}{L^p} \st \exists \acs{x_k} \subset \R^d, v_{n_k}\pa{\cdot - x_k} \us{L^p}{\wra} v} \sle 1/c\ind{CLR}.
\end{align*}
Then by the CLR bound \cite{Cwikel77,Lieb97,Roz72}, the remaining potentials to which subsequences can weakly converge, up to translations, are not able to bind any electron. Hence the system will split in adequacy with
\begin{align*}
\exc{0}(v_{n_k}) \lra \mymin{N_1, \dots, N_K \in \acs{0,\dots,N}\\\sum_{i=1}^K N_i = N} \sum_{i=1}^K \excn{0}{N_i}(v_i),
\end{align*}
and the ground wavefunctions would follow the binding subsystems.

Hellman-Feynman formulas \cite{Guttinger32,Hellman37,Feynman39} in the Gateaux sense can be derived by using a standard method \cite[Theorem 2.6]{LieSim77b} (see also \cite{MouRui88} and references therein for more general studies). The same method, used with the previously established regularity properties, enables to deduce a Hellmann-Feynman formula in the Fréchet sense.

\begin{corollary}[Hellmann-Feynman]\label{coco}
Let $p$ be as in \eqref{expot}, and choose $\spa = \spp$. The ground energy $v \mapsto \exc{0}(v)$ is Lipschitz continuous, concave and weakly upper semi-continuous on $\spp$. When $p > \max(2d/3,2)$, it is strictly concave and strictly increasing on $\spf{0}$. The energies $v \mapsto \exc{k}(v)$ are $\cC^{\ii}$ on $\spn{k}$, and for all $u \in \spp$,
\begin{align*}
\pa{\d_v \exc{k}} u = \int_{\Omega} u \ro_{\p\ex{k}(v)}.
\end{align*}
\end{corollary}

The previous expression can be formally written $\bpa{\d_v \exc{k}}^* = \ro_{\p\ex{k}(v)}$, where $\tx{}^*$ denotes the dual representation at stake in Riesz' theorem, and this corresponds to the notation $\restr{\f{\delta \exc{k}}{\delta u(x)}}{v}$ used in the physics litterature. The second differential is $\bpa{\d^2_v \exc{k}}(h,u) =\ps{ h, \bpa{\d_v \bpa{s \mapsto \ro_{\p\ex{k}(s)}} } u}$.

\subsection{Singularities on degenerate potentials}\label{sc3}
We want to study the map $v \mapsto \exc{k}(v)$ on singular potentials $\spf{k} \backslash \spn{k}$, to complete our general picture. To this purpose, we use the Rayleigh-Schrödinger series, that is the power series in $\alpha \in \R$ of $\p\ex{k}(v+ \alpha u)$ and $\exc{k}(v + \alpha u)$. We need to define a slightly weaker version of Gateaux derivation, because the ground state of $H+\alpha G$ in a neighborhood of $0^+$ is in general different from the one at $0^-$. Take a manifold $X$ locally modelled on a real vector space $Y$. We say that a function $f : X \ra \R$ is Dini differentiable at $x \in X$ if for any direction $y \in Y$, 
\begin{align}\label{dinidef}
	\lim_{0 \le t \ra 0^+} \pa{f(x+ty)-f(x)}/t =: \pa{{^+}\delta_x f}(y)
\end{align}
exists, i.e. $f$ has an upper Dini derivative in every direction. We also define ${^-} \delta_x f (y) \df - \pa{{^+}\delta_x f}(-y)$. Higher Dini derivatives $^+\delta_x^n f(y)$ are defined similarly. For possibly degenerate potentials $v \in \spf{k}$, we consider the real sphere of real eigenstates 
\begin{align}\label{bloch}
	\cD\ex{k}(v) \df \acs{ \p \in \Ker \bpa{\hn(v)-\exc{k}(v)} \st \p(X) \in \R, \mediumint \ab{\p}^2 = 1}.
 \end{align}
We define the integers $m_k$ and $M_k$ by
\begin{align}\label{cdff}
	 \exc{m_k - 1}(v) \sle \exc{m_k}(v) = \dots = \exc{k}(v) = \dots = \exc{M_k}(v) \sle \exc{M_k+1}(v),
 \end{align}
	so $\dim \cD\ex{k}(v) = M_k-m_k+1$, with $\exc{-1}(v) \df -\ii$ by convention. 

\begin{theorem}[Degenerate Hellman-Feynman]\label{helf}
	Let $p$ be as in \eqref{expot}, $\cV = \spp$, $w \in \spp$, take a possibly degenerate potential $v \in \spf{k}$, and consider the definitions \eqref{dinidef}, \eqref{bloch} and \eqref{cdff}. The energy $\exc{k}$ is infinitely Dini differentiable on $v$, with
\begin{align}\label{degdeg}
^+\delta_v \exc{k} (u)  & =  \mymax{\p_0,\dots,\p_{M_k-k} \in \cD\ex{k}(v) \\ \p_i \perp \p_j \\ 0\le i,j \le M_k-k}  \mymin{\p = \sum_{i=0}^{M_k-k} \lambda_i \p_i\\ \lambda_i \in \C, \sum_i \ab{\lambda_i}^2 = 1} \int \ro_{\p} u   \nonumber\\
& =  \mymin{\p_0,\dots,\p_{k-m_k} \in \cD\ex{k}(v) \\ \p_i \perp \p_j\\ 0\le i,j \le k-m_k}  \mymax{\p = \sum_{i=0}^{k-m_k} \lambda_i \p_i\\ \lambda_i \in \C, \sum_i \ab{\lambda_i}^2 = 1} \int \ro_{\p} u.
\end{align}
In particular, ${^+}\delta_v \exc{0}(u) = \min_{\p \in \cD\ex{0}(v)} \int \ro_{\p} u$. Moreover, $^+\delta_v \exc{0}$ is concave and if $w \ge0$, it is also weakly upper semi-continuous.
	If in \eqref{cdff} we take $M_k = m_k +1$ and $k=m_k$, so that $\dim \cD\ex{k}(v) = 2$ and if $\p,\Phi$ is an orthonormal basis of $\cD\ex{k}(v)$, then we have
\begin{align}\label{coucou}
	& ^{\pm}\delta_v \exc{k}(u) = \ud\int u \pa{ \ro_{\p} + \ro_{\Phi}} \mp \ud \sqrt{ \pa{ \int \hspace{-0.1cm} u\pa{ \ro_{\p} - \ro_{\Phi}}}^2 \hspace{-0.2cm} + 4 \ab{ \ps{ \p, \vv{u} \Phi}}^2}.
\end{align}
\end{theorem}
Similar properties hold for $\p\ex{k}$, which is infinitely Dini differentiable. If $^+\delta_v \exc{k}(u) \sle \tx{ } ^-\delta_v \exc{k}(u)$, then the perturbation of $\hn(v)$ by $u$ decreases the degeneracy by at least one, since locally at least two eigenvalues split. The degeneracy of $\cD\ex{0}(v)$ is completely broken at first order if and only if $\min_{\p \in \cD\ex{0}(v)} \int \ro_{\p}u$ has a unique minimizer up to a phase factor, because then, locally there is exactly one eigenvalue becoming lower than all the others. Given a real orthonormal basis $(\p_i)_{1 \le i \le D}$ of $\Ker \bpa{\hn(v)-\exc{k}(v)}$, we can parametrize $\p = \sum_{i=1}^D \lambda_i \p_i$ with complex $\lambda = (\lambda_i)_{1 \le i \le D}$ satisfying $\sum_{i=1}^D \ab{\lambda_i}^2 = 1$, and we have $\int u\ro_{\p} = \ps{\lambda, M_u \lambda} \in \R$ where $M_u \df \bpa{\int u(x_1) \p_i \p_j}_{1 \le i,j \le D}$ is symmetric and real.





Given some degenerate potential $v \in \spf{k} \backslash \spn{k}$, we want now to know whether there is a direction in which one can break the degeneracy. We also want to know whether $\exc{k}$ is differentiable at those degenerate potentials.

\begin{corollary}[Degeneracy breaking and differentiability of $\exc{k}$]\label{propcandif} Let $p$ be as in \eqref{expot}, $\spa = \spp$, $w \in \spp$ and $v \in \spf{k}$, and consider $\cD\ex{k}(v)$ as defined in \eqref{bloch}. Assume that $k=0$, that $\exc{k-1}(v) \sle \exc{k}(v)$, or that $\exc{k}(v) \sle \exc{k+1}(v)$.
\smallskip

\textup{($i -$ Breaking in a direction.)} Take a direction $u \in \spp$. The following statements are equivalent

\hspace{1cm}\bul $\lambda \mapsto {^+}\delta_v \exc{k}(\lambda u)$ is linear on $\R$

\hspace{1cm}\bul the degeneracy is not broken at first order in the direction $u$

\hspace{1cm}\bul the integral $\int u \ro_{\p}$ is constant over $\p \in \cD\ex{k}(v)$

\hspace{1cm}\bul for any $\p,\Phi \in \cD\ex{k}(v)$ we have $N \int_{\R^{dN}} u(x_1) \p \Phi = \ps{\p,\Phi} \int u \ro_{\p}$

\textup{($ii -$ Generic breaking.)} The following statements are equivalent

\hspace{1cm}\bul $\exc{k}$ is differentiable at $v$

\hspace{1cm}\bul $\lambda \mapsto {^+}\delta_v \exc{k}(\lambda u)$ is linear on $\R$, for any $u \in \cC^{\ii}\ind{c}$

\hspace{1cm}\bul the degeneracy is never broken at first order, in any direction

\hspace{1cm}\bul the density $\ro_{\p} =: \ro$ is constant over $\p \in \cD\ex{k}(v)$

\hspace{1cm}\bul for any $\p,\Phi \in \cD\ex{k}(v)$, $N \int_{\R^{d(N-1)}} \p \Phi = \ps{\p,\Phi} \ro_{\p}$ 

\textup{($iii -$ The energy is not differentiable when $w=0$.)} Let $v \in \spf{k} \backslash \spn{k}$ be a degenerate potential, and $w=0$, let $\ell$ be the smallest $j \in \N$ such that $\exc{0}(v)\sle \exc{j}(v)$. If $k$ is such that $\exc{k}(v) \in \acs{\exc{0}(v),\exc{\ell}(v)}$, then $\exc{k}$ is not differentiable at $v$. If $N=1$, $\exc{k}$ is not differentiable at $v$ for any $k$.
\end{corollary}
In the case $iii)$ we have $^+\delta_v \exc{k}(u) \sle \tx{ }^-\delta_v \exc{k}(u)$ for at least one direction $u \in \spp$. We conjecture that at those degenerate potentials, $\exc{k}$ is not differentiable in the interacting case either, that is, there is a direction in which the left and right derivatives are different. The constraints that have to be satisfied to not break degeneracy at first order are strong. We think that the ground degeneracies are generically broken at some order and even that $\spn{k}$ is dense in $\spf{k}$.

\section{Main results: consequences for the inverse problem}

From the weak-strong continuity of $\p\ex{0}$, we can deduce negative results about the inverse continuity. We define the potential-to-ground state density map
\begin{align}\label{defro}
\ro : \begin{array}{rcl}
\spn{0} & \lra & W^{1,1}\pa{\Omega,\R_+} \cap \acs{ \int \cdot = N} \\
	v & \longmapsto & \ro(v) \df \ro_{\p\ex{0}(v)}. \\
\end{array}
\end{align}
It can also be defined on $\spf{0}$ as a multivalued map. The space $W^{\ell,1}(\Omega) \cap \acs{ \int_{\Omega} \cdot = N}$ is a closed embedded submanifold of $W^{\ell,1}(\Omega)$, hence a smooth manifold. 
 The main property of $v \mapsto \ro(v)$, lying at the heart of DFT, is its injectivity (when $\spa = (\spp)/ \sim$ instead of $\spa = \spp$), this is the Hohenberg-Kohn theorem proved in \cite{HohKoh64,Lieb83b,Garrigue19}, when $p > \max(2d/3,2)$. 

In 1965, Kohn and Sham postulated the existence of effective one-body potentials, removing the electronic interaction while keeping the same ground state density \cite{KohSha65}, by adding a one-body potential. The resulting non-interacting problem $\sum_{i=1}^N -\Delta_i + v\ind{ks}(x_i)$ is then much easier to study than $\hn(v)$. We will also denote by $\ro$ the multivalued density map defined on $\spfn{0}{N}$, and by $\ro\iv$ its inverse, which exists by the Hohenberg-Kohn theorem. Let us denote by $\ro_{w=0}$ the map $\ro$ for which $w=0$, then $\ro_{w=0}(\cV^{(0)}_{\partial,N,w=0})$ is the set of non-interacting potential-representable densities. Considering elements
 \begin{align*}
v \in \ro\iv\pa{ \ro(\spf{0}) \cap \ro_{w=0}(\spfn{0}{N,w=0})} = \spf{0} \cap \ro\iv \rond \ro_{w=0} \pa{ \spfn{0}{N,w=0}},
 \end{align*}
which is possibly empty, the Kohn-Sham potential is defined as
 \begin{align*}
v\ind{ks}(v) \df \ro_{w=0}\iv \rond \ro(v).
 \end{align*}
As wanted, $\ro(v) = \ro_{w=0}(v\ind{ks}(v))$. Knowing $\ro(\spf{0}) \cap \ro_{w=0}(\spfn{0}{N,w=0})$, as raised by Lieb in \cite[Question 8]{Lieb83b}, is thus an important open problem. 
 In the case that the Fermi level of $v$ for $w=0$ is filled, the map $\ro\iv$ also enables to express the self-consistent field (SCF) equations, which are two equivalent fixed-point relations fulfilled by (resp.) potentials in $\cV^{(0)}_{N,w=0}$ and densities in $\ro_{w=0}\bpa{\cV^{(0)}_{N,w=0}}$. They are formally written
\begin{align*}
v = \ro_{w=0}\iv\pa{x \mapsto \indic_{-\Delta +v \le\ep\ind{F} }(x,x)}, \bhs \ro = x \mapsto \indic_{-\Delta + \ro_{w=0}\iv(\ro) \le \ep\ind{F}}(x,x),
\end{align*}
where the Fermi level ${\ep\ind{F} \in \R}$ is such that only the first $N$ orbitals are taken, and $\indic_{A}$ denotes the spectral projection of a self-adjoint operator $A$. 

  The direct problem $\ro$ is well-posed in the standard sense \cite{Hadamard32,CouHil08,Hadamard64} by injectivity and regularity, and the Kohn-Sham problem is its corresponding inverse problem. 

 The linearization of this inverse problem is ill-posed because $\ro$ has a compact differential by Theorem \ref{maa}, which indicates the problematic nature of the existence of Kohn-Sham potentials. In bounded domains, the Kohn-Sham problem is ill-posed in the sense of Hadamard \cite[Definition p8]{HasAleRom17}, because $\ro\iv$ is discontinuous.

\begin{corollary}[The set of $v$-representable densities is topologically small]\label{smallcor} Let $p$ be as in \eqref{expot}, with $p > d/2$ when $d \ge 3$, and consider $\ro$ as defined in \eqref{defro}. When the system lives in a bounded open connected set $\Omega \subset \R^d$, then $v \mapsto \ro(v)$ is compact, its inverse $\ro\iv$ is discontinuous, and $\ro(\spn{0})$ is a countable union of compact sets. In particular, $\ro(\spn{0})$ has empty interior in $W^{1,1} \cap \acs{\int \cdot = N}$.
\end{corollary}

By Corollary \ref{smallcor}, $\ro_{w=0}(\sps{0}{N,w=0}) \cap \ro(\spn{0})$ is included in a countable union of compact sets, hence it is meagre in the sense of Baire. The Kohn-Sham potential thus seems to be defined on a sparse set, possibly empty, under our conditions on $p$. The situation cannot be much better when $\Omega$ is unbounded. 


Despite the previous negative results, we can still prove a weak inverse continuity property. 

\begin{proposition}[Weak inverse continuity of $\p\ex{k}$]\label{propru}
	Let $p > \max(2d/3,2)$. Let $v_n \in \spfn{k}{N}$ be a sequence of potentials such that $v_n - \exc{k}(v_n)/N$ is bounded in $\spp$. Take normalized eigenstates $\psi\ex{k}(v_n) \in \Ker \bpa{\hn(v_n) - \exc{k}(v_n)}$ such that $\psi\ex{k}\pa{v_n} \ra \psi\ex{k}(v)$ strongly in $H^2(\R^{dN})$ for some $v \in \spn{k}$ and some normalized $\psi\ex{k}(v) \in \Ker \bpa{\hn(v) - \exc{k}(v)}$. Then we can deduce that 
	\begin{align}\label{reelm}
	 \int_{\Omega} \pa{v_n-v- \f{\exc{k}(v_n)-\exc{k}(v)}{N}  }^2 \ro_{\psi\ex{k}(v)} \ra 0
 \end{align}
and $v_n  \ra v$ a.e. in $\Omega$, up to a constant and a subsequence.
\end{proposition}
The relation \eqref{reelm} implies that for any $c > 0$ and on level sets $X_c \df \{x \in \Omega \st \ro_{\p\ex{k}(v)}(x) \ge c \}$, we have $v_n \ra v + a$ strongly in $L^2(X_c)$ for some constant $a$. Since $|\{x \in \Omega \st \ro_{\psi\ex{k}(v)}(x)=0 \}|=0$ by unique continuation \cite{Garrigue19}, then $X_c$ will  \apo{approach} $\Omega$ as $c \ra 0$.

\section{Proof of Theorem \ref{bindos}}\label{secpot} 
We recall that the natural norm of $L^p+L^{\ii}$ is
\begin{align}\label{ddef}
\nor{v}{L^p+L^{\ii}} = \mymin{f \in L^p, g \in L^{\ii} \\ f + g = v} \bpa{ \nor{f}{L^p} + \nor{g}{L^{\ii}}}.
\end{align}
The set $\sppe$ is a closed subspace, it is the closure of $L^p$ in $\spp$.

\begin{remark}
	If in the definition of $\spf{0}$, we replace $\spp$ by $\sppe$, then $\cap_{n=1}^N \spfn{0}{n}$ is path-connected as well, as can be seen from our proof.
\end{remark}

\begin{remark}
As proved in \cite[Thm 3.11, Thm 3.12, Thm 3.13]{Lieb83b} the set of binding potentials $\spf{0} \cap L^p$ is dense in $L^p$. We can see it by approaching $v \in L^p$ with a sequence $v_n = v-\sum_{i=1}^ N L_n \indic_{B_{r_n}(y_i^n)}$ where $L_n$ and $r_n$ are chosen such that $v_n \in \spf{0}$, $0 \le L_n \ra 0$, $r_n \ra +\ii$, $y_i^n \in \R^d$, $\ab{y_i^n} \ra +\ii$, $|y_i^n - y^n_j|\ra +\ii$. This result also holds in $\sppe$ by density of $L^p$ in this space. 
\end{remark}

\begin{remark}
	Theorem \ref{bindos} raises the question of path-connectedness of $\spn{0}$. Adiabatic processes are deformations of the potential when the initial system is in its ground state, slowly enough so that the system remains in the ground state thanks to the adiabatic theorem \cite{Kato50,Teufel03}. The time scale of change in $v$ needs to be small with respect to the energy difference between the first two levels, hence a necessary and sufficient condition for this process to be possible is to remain in $\spn{0}$ during the deformation, without crossing the degenerate potentials $\spf{0} \backslash \spn{0}$, otherwise excited states can be populated. By analogy with other areas of quantum physics we say that two potentials $v,u \in \spn{0}$ are \textit{adiabatically equivalent} if they are path-connected in $\spn{0}$. This defines equivalence classes in $\spn{0}$ and it would be interesting to know whether there is only one class. We remark that in classical mechanics this is the case. Graphically, degenerate potentials $\spf{0} \backslash \spn{0}$ constitute a \apo{web} in the space of binding potentials. 
\end{remark}

\begin{remark}
	The proof of Theorem \ref{bindos} does not extend to the case of excited states because we use the HVZ theorem, which only involves ground energies.
\end{remark}

We now prepare for the proof of Theorem \ref{bindos}. The following lemma will allow us to modify potentials while remaining bound.

\begin{lemma}\label{addos}
If $v \in \spf{0}$, $0 \le u \in \spp$ and 
\begin{align}\label{infker}
\mymin{\p \in \Ker \pa{\hn(v) - \exc{0}(v)} \\ \int \ab{\p}^2 = 1} \int u \ro_{\p} \sle \sign(v) - \exc{0}(v),
\end{align}
	then $\exc{0}(v+u) \sle \sign(v+u)$.
\end{lemma}
\begin{proof}
	By the min-max theorem, $\sign(v) \le \sign(v+u)$. Let $\p_v$ be one ground state minimizing the left-hand side of \eqref{infker}, which is a minimization problem in a compact set since $\dim \Ker \bpa{\hn(v)-\exc{0}(v)} \sle +\ii$. We compute
\begin{align*}
\exc{0}(v+u) \le \cE_{v+u}(\p_v) = \exc{0}(v) + \int u \ro_{\p_v} \sle \sign(v) \le \sign(v+u),
\end{align*}
and consequently $v+u \in \spf{0}$.
\end{proof}


\begin{proof}[Proof of Theorem \ref{bindos}] We split the proof into several steps. Consider a binding potential $v \in \cap_{n=1}^N \spfn{0}{n}$. We will deform it continuously into a hole of type $-c \indic_{B}$, such that it remains in $\cap_{n=1}^N \spfn{0}{n}$ during the deformation. \\

	\textit{Step 1: connection to a negative bounded potential with compact support.} We decompose $v = v_p + v_{\ii}$ where $v_p \in L^p$ and $v_{\ii} \in L^{\ii}$. We start by transforming $v_p$ to $v_p \indic_{\ab{v_p} \le M}$, and when $M$ is large enough, this operation changes infinitesimally $\exc{0}(v)$ and $\sign(v)$ by classical perturbation theorems \cite{Kato,ReeSim4}, so $\sign - \exc{0}$ remains strictly positive during the modification. More precisely, we can take
 \begin{align*}
 v(t) \df (1-t) v + t \pa{v_p \indic_{\ab{v_p} \le M} + v_{\ii}} = v_{\ii} + v_p\indic_{\ab{v_p} \le M} +(1-t) v_p \indic_{\ab{v_p} > M}  
 \end{align*}
which links $v$ to $v_1 \df v_p \indic_{\ab{v_p} \le M} + v_{\ii} \in L^{\ii}$ by a line on which $v(t) \in \cap_{n=1}^N \spfn{0}{n}$ for all $t \in \seg{0,1}$.

	Let us denote by $\p^n_{v_1}$ a ground state of $v_1$, for all $n \in \acs{1, \dots,N}$. We take some $L \ge \nor{v_1}{L^{\ii}}$, and consider the path of positive potentials $u(t) = t (L-v_1)\indic_{\R^{d} \backslash B_r} \ge 0$ for $t\in\seg{0,1}$. We know that $\p^n_{v_1}$ decays exponentially at infinity \cite{Agmon,Simon82}, so we can choose $r = r(L,v_1,w,N)$ large enough such that 
\begin{align*}
\mysup{t \in \seg{0,1}} \int u(t)\ro_{\p^n_{v_1}} \le \bpa{L+ \nor{v_1}{L^{\ii}}} \int_{\R^{d} \backslash B_{r}} \ro_{\p^n_{v_1}} \sle \sign(v_1) - \exc{0}(v_1),
\end{align*}
	for any $n \in \acs{1, \dots, N}$. Hence by Lemma \ref{addos}, $v_1 + u(t) \in \cap_{n=1}^N \spfn{0}{n}$ for any $t \in \seg{0,1}$, and we redefine 
 \begin{align*}
 v_2 \df v_1 +u(1) = L \indic_{\R^{d} \backslash B_r} + v_1 \indic_{B_r},
 \end{align*}
for the rest of the proof. We then use that $\sign(v_2) - \exc{0}(v_2)$ is invariant under the gauge transformation $v_2 \ra v_2 + c$ so we move the potential by adding $-tL$, $t \in \seg{0,1}$. We \apo{filled $v_2$ up to the roof} and obtain 
 \begin{align*}
v_3 \df v_2 - L = (v_1-L) \indic_{B_r} \le 0.
 \end{align*}
 We have thus linked our potential to a negative potential with compact support. \\

	\textit{Step 2: build the wall.} Next we raise some big wall again, further away. We want to apply Lemma \ref{addos} to $u(t) = t\ell \indic_{\R^d \backslash B_R}$. We choose $R \in \R$ with $R\ge \max(r, \diam \supp v_3)$ and $\ell \ge 1$ so that
\begin{align}\label{lk}
\ell \int_{\R^d \backslash B_R} \ro_{\p^n_{v_3}} \le \sign(v_3) - \exc{0}(v_3)
\end{align}
	where $\p^n_{v_3}$ is one of the ground states of $v_3$. We know that there exist $\alpha, \beta > 0$ such that $\int_{\R^d \backslash B_R} \ro_{\p^n_{v_3}} \le \alpha e^{-\beta R}$ \cite{Simon82,Agmon,Hislop00}, hence we link $\ell$ and $R$ by taking 
	\begin{align}\label{link}
R = c (1+\ln \ell),
 \end{align}
	with $c$ large enough so that \eqref{lk} holds. In all the following steps, if our properties hold for some pair $R$, $\ell$ large enough, such as described, then they hold for any pair $R'$, $\ell'$ where $\ell' \ge \ell$ and $R' \ge c (1+\ln \ell')$. By Lemma \ref{addos} we deduce that $t\ell \indic_{\R^d \backslash B_R} +v_3 \in \cap_{n=1}^N \spfn{0}{n}$ for any $t \in \seg{0,1}$. In particular, $\ell \indic_{\R^d \backslash B_R} +v_3 \in \cap_{n=1}^N \spfn{0}{n}$. Shifting this last potential by the constant $-t\ell$ for $t \in \seg{0,1}$, we also have 
\begin{align*}
\ell (1 - t) \indic_{\R^d \backslash B_R} + (v_3-t\ell) \indic_{B_R} \in \cap_{n=1}^N\spfn{0}{n}. 
\end{align*}
	In particular, $(v_3-\ell) \indic_{B_R}  \in \cap_{n=1}^N\spfn{0}{n}$. We can hence choose $\ell$ as large as we want, and $R(\ell)$ will also be large.\\

	\textit{Step 3: seal the hole.} Next we define $V_{\ell,R} \df \pa{v_3-\ell} \indic_{B_R}$, for any $\ell \ge 1$ and $R = R(\ell)$ linked by \eqref{link}. By the HVZ theorem \cite{Zhislin60b,VanWinter64,Hun66} used in the form of \cite[Theorem 3.1]{Lewin11}, we have $\Sigma_n(V_{\ell,R}) = \excn{0}{n-1}(V_{\ell,R})$ for any $n \ge 1$, $\ell \ge 0$ (with the convention $\excn{0}{0} \df 0$).   

	Take $a > 0$ fixed and let us denote by $(\vp_i)_{1\le i \le N}$, $\vp_i \in H^1(B_a,\C)$ an orthonormal familly of functions. For $R \ge a$ and for any $n \in \segint{1}{N}$ we have
 \begin{align*}
	 -\ell n \le \exc{0}(-\ell \indic_{B_R}) \le \cE_0\pa{\wedge_{i=1}^n \vp_i} - \ell n,
 \end{align*}
	and we deduce that $\exc{0}(-\ell \indic_{B_R}) = -\ell n + O(1)$ when $\ell \ra +\ii$. Since $\nor{v_3}{L^{\ii}} \le M$ by the first step, then
 \begin{align*}
	 & \exc{0}(V_{\ell,R}) - \signn{n}(V_{\ell,R}) = \exc{0}(V_{\ell,R}) - \excn{0}{n-1}(V_{\ell,R})\\
	 & \bhs \le \exc{0}\pa{- \ell \indic_{B_R}} - \excn{0}{n-1}\bpa{- (\ell + \nor{v_3}{L^{\ii}}) \indic_{B_R}} \\
	 & \bhs \le \cE_0\pa{\wedge_{i=1}^n \vp_i} - \ell n + \pa{\ell + \nor{v_3}{L^{\ii}}} (n-1) \\
	 & \bhs \le - \ell + c_{M,w,N},
 \end{align*}
where $c_{M,w,N}$ does not depend on $\ell$ or on $R$.
 We thus showed that 
\begin{align*}
\mymax{n \in \acs{1,\dots,N}} \bpa{\exc{0}(V_{\ell,R}) - \signn{n}(V_{\ell,R})} \ra - \ii
\end{align*}
when $\ell\ra +\ii$. We take $\ell$ large enough so that 
\begin{align}\label{equ}
  \int v_3 \ro_{\p^{n}(V_{\ell,R})} \le M N \sle \ell - c_{M,w,N} \sle \signn{n}(V_{\ell,R}) - \exc{0}(V_{\ell,R}),
\end{align}
for any $n \in \acs{1, \dots , N}$. 

Then again applying Lemma \ref{addos} to $u(t) = -tv_3 \ge 0$, and using \eqref{equ}, we have $V_{\ell,R} - tv_3 \in \cap_{n=1}^N \spfn{0}{n}$ for any $t \in \seg{0,1}$. In particular, $V_{\ell,R} -v_3 = -\ell\indic_{B_R}\in \cap_{n=1}^N \spfn{0}{n}$. \\

\textit{Step 4: connect any two potentials.} We showed how to connect an initial binding potential $s$ to one well $-\ell_s \indic_{B_{R_s}}$, where $R_s = c_s(1+\ln \ell_s)$, and we know that $s$ it is still connected to $-\ell \indic_{B_{R}}$ for any $\ell \ge \ell_s$ and any $R \ge c_s (1+\ln \ell)$. To connect $v$ to another binding potential $u$, we can connect $v$ to $-\ell_v \indic_{B_{R_v}}$ where $R_v = \max(c_u,c_v) (1+\ln \ell_v)$, $u$ to $-\ell_u \indic_{B_{R_u}}$, where $R_u = \max(c_v,c_u) (1 + \ln \ell_u)$. Then we know that $v$ and $u$ are connected to $-\ell \indic_{B_{R}}$ for any $\ell \ge \max(\ell_v,\ell_u)$ and any $R \ge \max(c_v,c_u) (1+ \ln \max(\ell_u,\ell_v))$ so we can connect $v$ and $u$ to common wells.
\end{proof}

 \begin{proof}[Proof of Corollary \ref{rkb}.] We show here how to find a continuous path of ground states which links any initial and final ground states corresponding to the path of potentials in the previous proof. We follow an argument used in the proof of \cite[Theorem 4]{Lewin04b}. We take a ground state $\p_0$ of the initial potential $v = v(0)$ and a ground state $\p_1$ of the final potential well $-c \indic_{B_R} = v(1)$. We consider the previous path of potentials $t \in [0,1] \mapsto v(t)$, which is piecewise linear. 

We now use a result due to Rellich \cite[Theorem 1.4.4, Corollary 1.4.5]{Simon15}, showing that for a Hamiltonian depending on one parameter, the singularities corresponding to eigenvalues crossings can be \apo{removed} in the sense that in the neighborhood of any crossing, there exist analytic branches representing those eigenvalues. So let us denote by $(\vp_i(t))_{i \in D(t)}$ an orthonormal basis of $\Ker \bpa{\hn(v)-\exc{0}(v)}$, which by Rellich's theorem can be chosen analytic on the left and on the right at each point, with possibly different limits. There is only a finite number of jumps. Indeed, the essential spectrum is strictly separated from $\exc{0}(v(t))$, uniformly in $t \in \seg{0,1}$, thus only a finite number of eigenfunctions are involved. So if the left and right limits are different an infinite number of times, this is because two eigenfunctions $\phi_j(t)$ and $\phi_{\ell}(t)$ cross an infinite number of times at the ground state energy level. Since they are analytic, their energies must be equal and hence their energies \textit{a posteriori} do not cross.

By the above argument, we have $k \in \N$ and $t_0, \dots, t_k \in \seg{0,1}$ such that $(\vp_i(t))_{i \in D(t)}$ is piecewise analytic on $\seg{t_i,t_{i+1}}$. On $]t_i,t_{i+1}[$, we choose a ground eigenfunction path $\p(t) \in \ran (\vp_i(t))_{i \in D(t)}$. When there is a crossing of eigenvalues at $t_i \in [0,1]$, by analyticity of the eigenfunctions, there are definite limits $(\vp_i(t^-_i))_{i \in D(t^-_i)}$ and $(\vp_i(t^+_i))_{i \in D(t^+_i)}$ on the left and on the right of $t_i$.  Let us denote by $\p(t_0^-)$ and $\p(t_0^+)$ the ground state limits on the left and on the right. At the interface, the ground eigenspace of $v(t_0)$ is $\ran(\vp_i(t^-_i))_{i \in D(t^-_i)}  + \ran (\vp_i(t^+_i))_{i \in D(t^+_i)}$. The set of normalized ground eigenstates is the unit sphere of this vector space, and we can add a path of ground eigenfunctions staying on this sphere to connect $\p(t_i^-)$ and $\p(t_i^+)$. 
\end{proof}

We conjecture that $\spfn{0}{N+1} \subset \spfn{0}{N}$. This is striking that such an intuitive fact is not direct to show. We also remark that the convexity of $N \mapsto \exc{0}(v)$ would imply it, this last statement being a conjecture as well \cite{PerParLev82,BacDel14}. A counterexample to the convexity of $N \mapsto \exc{0}(v)$ is given in a remark after \cite[Theorem 4.1]{Lieb83b} when the interaction $w$ is a soft core and for $d=3$. But in this case $\spfn{0}{N+1} \subset \spf{0}$ still holds, so we conjecture that $\spfn{0}{N+1}\subset \spfn{0}{N}$ holds for any interaction $w \ge 0$.

It also seems natural that if $w = \ab{\cdot}\iv$, $v,u \in \spa$, $v \le u$ and $u \in \spn{0}$, then $v \in \spn{0}$. We hence conjecture that $v \mapsto \excn{0}{N+1}(v) - \excn{0}{N}(v)$ is increasing on $\spfn{0}{N} \cap \spfn{0}{N+1}$ for the Coulomb interaction. This could be a special property of this interaction, because as we mentioned before, $N \mapsto \exc{0}(v)$ is not convex for soft core interactions for instance. 

\section{Proofs: the wavefunction-to-projector map}\label{wtp} 
In this section we present several basic facts about the space of orthogonal projectors $\acs{ \proj{\p} \st \p \in H^1(\Omega^N), \nord{\p}^2 =1}$ and the map $\p \mapsto \proj{\p}$.

\subsection{Main properties}

Quantum pure states are rays of projective Hilbert spaces \cite[Section 2.1]{Weinberg96}. By nature, the map $v \mapsto \proj{\p(v)}$ has no information on the phase of the ground states, so we will adapt the projective approach to regular pure states. We denote by 
 \begin{align*}
	  \Sbb \df \acs{ \p \in L\ind{a}^2(\R^{dN}) \st \nord{\p} = 1}, \bhs  H^k\ind{p} \df \frac{H^k \cap \Sbb}{ S^1},
 \end{align*}
respectively the unit sphere of the set of antisymmetric wavefunctions $L\ind{a}^2(\R^{dN})$, and the Sobolev spaces corresponding to physical wavefunction, where $S^1$ is the unit circle of dimension one representing the phase of a pure state. We denote by $[\cdot]$ the canonical projection of $H^k$ onto $H^k\ind{p}$. The indice \apo{p} can either stand for \apo{physical} or \apo{projective}. On this space, the natural metric is
 \begin{align*}
	 \D_k(\p,\Phi) & \df   \myinf{\psi, \phi \in H^k \cap \Sbb \\ \seg{\psi} = \p, \seg{\phi} = \Phi} \nor{\phi - \psi}{H^k} = \D\pa{ (-\Delta+1)^{\f{k}{2}} \p, (-\Delta+1)^{\f{k}{2}} \Phi},
 \end{align*}
 where
 \begin{align*}
	 \D(\p,\Phi)^2  \df \D_0(\p,\Phi)^2 = \nor{\p}{L^2}^2 + \nor{\vp}{L^2}^2  -2\ab{\ps{\vp,\psi}}.
 \end{align*}
In the case $k=0$, the main properties of these objects are well-known \cite{Pflaum19}, and we adapt them for $k\ge 1$. The next proposition shows that $H^k\ind{p}$ is a smooth manifold, on which one can use differential geometry.
\begin{proposition}\label{lka}
	The space $H^k\ind{p}$ is a completely metrizable (via $\D_k$) smooth manifold modelled on a Hilbert space isomorphic to each of the Hilbert spaces $\acs{\psi}^{\perp} \cap H^k$ where $\psi \in H^k$. Moreover, for any $\p \in H^k\ind{p}$, $\Td_{\p} H^k\ind{p} \simeq \acs{\psi}^{\perp} \cap H^k$ locally, where $[\psi] = \p$.
\end{proposition}

 We will denote by $\norop{\cdot}$ the operator norm. We define the $k^{\tx{th}}$ Sobolev space of operators $\chh$ as being the linear space of bounded self-adjoint operators $\cB\pa{L^2(\R^{dN})}$ which norm
 \begin{align*}
\nor{A}{\chh} \df \norop{ (-\Delta+1)^{\f{k}{2}} A (-\Delta+1)^{\f{k}{2}}}
 \end{align*}
 is finite, it is a Banach space. Similarly, we define the Sobolev-Schatten spaces $\sch_{k,p}$ on self-adjoint operators via their norms
 \begin{align*}
	 \nor{A}{\sch_{k,p}}^p \df \tr \ab{ (-\Delta+1)^{\f{k}{2}} A (-\Delta+1)^{\f{k}{2}}}^p.
 \end{align*}
 We then define the state-to-projector map
\begin{align*}
\pt :
\begin{array}{ccl}
H^k\ind{p} & \lra & \chh \cap \acs{\tr \cdot = 1} \cap \acs{\norop{\cdot}=1} \\ 
\p & \longmapsto & \proj{\p},
\end{array}
\end{align*}
and can show that it is very regular.
\begin{proposition}[$\pt$ is an embedding]\label{embedding}
$\pt$ is a smooth embedding, $\pt\iv$ is globally H\"older and $\cC^{\ii}$.
\end{proposition}
For the definition of an embedding, see \cite[p559]{Zeidler4}. As a corollary of the previous results, the space $\im \pt$ is smooth.
 \begin{corollary}\label{cc}
	 $\im \pt$ is a submanifold of $\chh \cap \acs{\tr \cdot = 1} \cap \acs{\norop{\cdot}=1}$, $\tx{T}_{\pt(\p)} \im \pt = \im \d_{\p} \pt$, and all the topologies $(\sch_{p,k})_{p \in [1,+\ii]}$ on $\chh$ are equivalent.
\end{corollary}

\subsection{Proofs of Propositions \ref{lka} and \ref{embedding}}

In the literature, the case $H^0\ind{p}$ is studied in \cite{Pflaum19}. Its natural inner product is the projective inner product 
 \begin{align*}
	 \pa{[\psi],[\vp]} \df \f{\ab{\ps{\psi,\vp}}}{\nord{\psi}\nord{\vp}}.
 \end{align*}
The space $H^k\ind{p}$ is in bijection with $\PP H^k$, but we will not directly endow it with the same structure as in general projective Hilbert spaces theory. Indeed, in this case the metric would be 
 \begin{align*}
	 \pa{\p,\Phi} \mapsto \myinf{ [\psi] = \p, [\phi] = \Phi \\ \nor{\psi}{H^k} = \nor{\phi}{H^k} = 1} \nor{\psi - \phi}{H^k}
 \end{align*}
 but it is not the one we want to work with. The relevant one is $\D_k$. On $H^{\ell} \times H^{\ell}$, we have $\D_k \le \D_{\ell}$ for $k \le \ell$. A property of $\D$ is that 
 \begin{align}\label{ici2}
	 \ab{ \nord{\p}^2 - \nord{\Phi}^2} \le \D(\p,\Phi)^2.
 \end{align}
Moreover, 
 \begin{align*}
	 \D_k\pa{[\psi],[\vp]} & = \myinf{\theta \in \segod{0,2\pi} } \nor{(-\Delta+1)^{\f{k}{2}} \pa{\vp - e^{i\theta} \psi}}{L^2} \\
	 & = \sqrt{\nor{\psi}{H^k}^2 + \nor{\vp}{H^k}^2  -2\ab{\ps{(-\Delta+1)^{\f{k}{2}} \psi,(-\Delta+1)^{\f{k}{2}} \vp}}} \\
	 & = \D\pa{\seg{(-\Delta+1)^{\f{k}{2}} \psi}, \seg{(-\Delta+1)^{\f{k}{2}} \vp}}.
 \end{align*}
 First we prove Proposition \ref{lka}.
\begin{proof}[Proof of Proposition \ref{lka}]\tx{ }

	\bul Let us denote by $\pi$ the canonical projection from $H^k \cap \Sbb$ onto $H^k\ind{p}$. For each unit vector $\vp \in H^k \cap \Sbb$ we define the open sets $U_{\vp} \df \pi\bpa{H^k \backslash \acs{\vp}^{\perp}} \subset H^k\ind{p}$. The charts $h_{\vp} : U_{\vp} \ra \acs{\vp}^{\perp}$ are defined by 
 \begin{align*}
	 h_{\vp}( \pi(\psi) ) \df \f{(1-P_{\vp})\psi}{\ps{\vp,\psi}} = \f{\psi}{\ps{\vp,\psi}} - \vp
 \end{align*}
	 for any $\psi \in H^k \backslash \acs{\vp}^{\perp}$, where $(U_{\vp})_{\vp \in H^k \cap \Sbb}$ covers $H^k\ind{p}$. Those charts are $\cC^{\ii}$, we can verify that they are also injective and that their inverses are the maps $\acs{\vp}^{\perp} \ra U_{\vp}, \psi \mapsto \pi(\psi + \vp)$, which are also $\cC^{\ii}$, hence $h_{\vp}$ are smooth diffeomorphisms. For $\vp, \psi \in H^k \cap \Sbb$, the transition maps
\begin{align*}
	 h_{\vp} \rond h_{\psi}\iv : 
\begin{array}{ccc}
	h_{\psi}\pa{ U_{\vp} \cap U_{\psi}} & \lra & h_{\vp}\pa{ U_{\vp} \cap U_{\psi}} \\
	\phi & \longmapsto & \f{(1-P_{\phi})(\vp+\psi)}{\ps{\vp+\psi,\phi}}
\end{array}
 \end{align*}
	 are $\cC^{\ii}$ by composition. More precisely, the proofs follow from \cite{Pflaum19}.

	 \bul $\D_k$ is positive and symmetric. Assume that for $\p, \Phi \in H^k\ind{p}$, $\D_k(\p,\Phi) = 0$. Then for given $\psi,\phi \in H^k \cap \Sbb$ such that $\seg{\psi} = \p$ and $\seg{\phi} = \Phi$, there exists a sequence $\theta_n \in \segod{0,2\pi}$ such that
 \begin{align*}
\nor{\psi-e^{i\theta_n} \phi}{H^k} \ltend{n \ra +\ii} 0.
 \end{align*}
	 Up to a subsequence, $\theta_n \ltend{n \ra +\ii} \theta \in \segod{0,2\pi}$, then $\nor{\psi-e^{i\theta} \phi}{H^k}$ and $\p = \Phi$. Now for $\p,\Phi,\Xi$ and $\xi \in H^k\ind{p}$ such that $\seg{\xi} = \Xi$, we have
 \begin{align*}
	 & \D_k\pa{ \Phi, \p } = \myinf{\psi, \phi \in H^k \cap \Sbb \\ \seg{\psi} = \p,\seg{\phi} = \Phi} \nor{\phi - \psi}{H^k} \le \myinf{\psi, \phi \in H^k \cap \Sbb \\ \seg{\psi} = \p,\seg{\phi} = \Phi} \bpa{ \nor{\phi - \xi}{H^k} + \nor{\xi - \psi}{H^k} } \\
	 & \bhs = \myinf{\phi \in H^k \cap \Sbb \\\seg{\phi} = \Phi}  \nor{\phi - \xi}{H^k} + \myinf{\psi\in H^k \cap \Sbb \\\seg{\psi} = \p}  \nor{\xi - \psi}{H^k} = \D_k\pa{ \Phi,\Xi} + \D_k\pa{\Xi,\p},
 \end{align*}
	 and we can conclude that $\D_k$ is a metric.
 \end{proof}
Next, our goal is to relate vectors in $H^k\ind{p}$ with their corresponding rank-one projectors. For $k=0$ \cite{Pflaum19}, $\pt$ is bi-Lipschitz, with constants
 \begin{align}\label{bili}
	 2^{-\ud} \D(\p,\Phi)  \le  \norop{ P_{\p} - P_{\Phi}} \le \D(\p,\Phi).
 \end{align}
For our application we will need to work at $k=1$. We first make some short preliminary computations.
 
 \begin{lemma}\tx{ }

	 $i)$ For any $\chi, \phi \in L^2$, $\norop{ \ketbra{\phi}{\chi} } = \nord{\phi} \nord{\chi}$.

	 $ii)$ If moreover $\chi \perp \phi$, then $\norop{ \ketbra{\chi}{\phi} + \ketbra{\phi}{\chi}} = \nord{\chi} \nord{\phi}$.
 \end{lemma}
 \begin{proof}
	 $i)$ We have
 \begin{align*}
	 \norop{ \ketbra{\phi}{\chi} } & = \mysup{ \xi \in L^2 \cap \Sbb} \nord{ \ketbra{\phi}{\chi} \xi} = \nord{\phi} \mysup{ \xi \in L^2 \cap \Sbb} \ab{\ps{\chi,\xi}} \\
	 & = \nord{\phi} \ab{\ps{\chi,\f{\chi}{\nord{\chi}}}} =\nord{\chi} \nord{\phi}.
 \end{align*}

	 $ii)$ We can compute the norm by using the equality
 \begin{align*}
	 \bpa{\ketbra{\chi}{\phi} + \ketbra{\phi}{\chi}}^2 = \nord{\chi}^2 \proj{\phi} + \nord{\phi}^2 \proj{\chi},
 \end{align*}
	 and the fact that $\phi \perp \chi$.
\end{proof}

 Now we establish our main estimates, relating the metric $\D_k$ with the $\chh$ norm on rank one projectors.

\begin{lemma}\label{eqqq} For any $\psi,\vp \in L^2$, we have
	\begin{align}\label{inq}
\bpa{ \tr \ab{ P_{\vp} - P_{\psi}}}^2 & = \pa{ \nord{\psi}^2 + \nord{\vp}^2}^2 -4 \ab{\ps{\psi,\vp}}^2 \nonumber \\
	& = \D\bpa{[\psi],[\vp]}^4 + 4\ab{\ps{\psi,\vp}} \D\bpa{[\psi],[\vp]}^2,
\end{align}
and
\begin{align*}
\norop{P_{\psi} - P_{\vp}} & = \ud  \tr \ab{P_{\psi} - P_{\vp}} +\ud\ab{\nord{\psi}^2 - \nord{\vp}^2}.
\end{align*}
\end{lemma}
\begin{proof}
	The operator $P_{\psi} - P_{\vp}$ has its eigenvectors of the form $\chi = \alpha \psi + \beta \vp$ for some $\alpha, \beta \in \C$. Hence the system $\pa{P_{\psi} - P_{\vp}} \chi = \lambda \chi$ can be written
 \begin{align*}
	 \begin{cases} 
		 \alpha \nord{\psi}^2 + \beta \ov{z} - \lambda \alpha = 0\\ 
	 \alpha z + \beta \nord{\vp}^2 + \lambda \beta = 0,
	 \end{cases}
 \end{align*}
	where $z \df \ps{\vp,\psi}$. We assume that $\psi \neq \vp$ (in $(L^2 \cap \Sbb) / S^1$), $z \neq 0$, $\alpha \neq 0$ and $\beta \neq 0$, because the same conclusions will hold in those cases. Expressing $\alpha$ using the second equation, replacing it in the first one, multiplying by $z$ and dividing by $\beta$, we obtain
 \begin{align*}
	 \lambda^2 + \lambda \pa{ \nord{\vp}^2 - \nord{\psi}^2} + \ab{z}^2 - \nord{\psi}^2 \nord{\vp}^2 = 0.
 \end{align*}
	The eigenvalues are thus $\lambda_{\pm} = \ud\pa{\nord{\psi}^2 - \nord{\vp}^2 } \pm \sqrt{\Delta}$ where
 \begin{align*}
	  	\Delta \df \pa{\nord{\psi}^2 + \nord{\vp}^2}^2 - 4 \ab{z}^2. 
 \end{align*}
Since $\ab{ \nord{\psi}^2 - \nord{\vp}^2 } \le \sqrt{\Delta}$, we have
\begin{equation}\label{klo}
\left\{
\begin{array}{l}
	 \norop{P_{\psi} - P_{\vp}}  = \max\bpa{ \ab{\lambda_-},\ab{\lambda_+}} = \ud\ab{\nord{\psi}^2 - \nord{\vp}^2} + \ud \sqrt{\Delta}, \\
	 \tr \ab{P_{\psi} - P_{\vp}}  = \ab{\lambda_-} + \ab{\lambda_+} = \sqrt{\Delta}.
\end{array}
\right.
\end{equation}
This implies the conclusion of the lemma.
\end{proof}

In particular, this lemma shows that 
 \begin{align}\label{tops}
	  \ud  \tr \ab{P_{\psi} - P_{\vp}} \le \norop{P_{\psi} - P_{\vp}} \le \tr \ab{P_{\psi} - P_{\vp}}
 \end{align}
and proves that the $\sch_{k,p}$ norms are all equivalent, for all $p \in \seg{1,+\ii}$, on the space $\acs{ P_{\p} \st \p \in H^k}$. This also implies
 \begin{align*}
	 \D\bpa{[\psi],[\vp]}^4 \le \pa{ \tr \ab{ P_{\vp} - P_{\p}}}^2 \le (1+\ep) \D\bpa{[\psi],[\vp]}^4 + \f{4\ab{\ps{\psi,\vp}}^2}{\ep}
 \end{align*}
for any $\ep > 0$. Finally, for any $c > 1$, and any $\psi, \vp \in H^k\ind{p}$,
\begin{align}\label{mmm}
	\ud \D_k\bpa{\p,\Phi}^2 \le \nor{ P_{\p} - P_{\Phi} }{\chh } \le (1+\ep) \D_k\bpa{\p,\Phi}^2 + \f{2}{\sqrt{c^2-1}}\nor{\p}{H^k} \nor{\Phi}{H^k}.
 \end{align}

We now consider the state-to-projector map $\pt$. We cannot directly work on $\chh \cap \acs{ \norop{\cdot} = 1}$ because this is not a manifold since $\norop{\cdot}$ is not differentiable, and we cannot choose $\chh \cap \acs{\tr \cdot = 1}$ either because we cannot prove it to be a manifold by applying the preimage theorem \cite[Theorem 73.C]{Zeidler4} since the trace norm is not controlled by the operator norm. Proposition \ref{embedding} states that $\im \pt$ has a convenient geometric structure on which we can use differential geometry without complications.

\begin{proof}[Proof of Proposition \ref{embedding}]\tx{ }

	\textit{\bul Regularity.} First, $\pt$ is injective. The map $\psi \mapsto \proj{\psi}$ from $H^1 \cap \Sbb$ to $\sch_{\ii,1}$ is $\cC^{\ii}$. Since $H^1\ind{p}$ is the quotient of $H^1 \cap \Sbb$ by the action of the proper and compact group $S^1$, then $\pt$ is also $\cC^{\ii}$. The tangent space of $H^k \cap \Sbb$ at some point $\psi \in H^k \cap \Sbb$ is
 \begin{align*}
\Td_{\psi}\bpa{ H^k \cap \Sbb} = H^k \cap \Td_{\psi} \Sbb = \acs{ \phi \in H^k \st \re \ps{\psi,\phi} = 0 },
 \end{align*}
and the tangent space of $H^k\ind{p}$ at some $\psi\in H^k\ind{p}$ is
 \begin{align}\label{ed}
	 \Td_{\pi(\psi)} H^k\ind{p} & \simeq  \pa{\Td_{\psi}\bpa{ H^k \cap \Sbb}} / \pa{ \Td_{\psi} (S^1 \cdot \psi)} \simeq \acs{ \phi \in H^k \st \ps{\psi,\phi} = 0 } \\ \nonumber
	 & = H^k \cap \acs{\psi}^{\perp}. 
 \end{align}
Finally, the differential of $\pt$, defined on each chart $U_{\seg{\psi}}$, is given by
\begin{align*}
\d \pt :
\begin{array}{ccc}
	U_{\seg{\psi}} \subset H^k\ind{p} & \lra & \cB\pa{ H^k\ind{p} \cap \acs{\psi}^{\perp} , \chh }  \\
	\seg{\psi} & \longmapsto & \bpa{\vp \mapsto \ketbra{\vp}{\psi} + \ketbra{\psi}{\vp} }. \\ 
\end{array}
\end{align*}
We can show that it is injective.


	\textit{\bul Splitting.} To show that $\pt$ is an immersion, it remains to show that for any $\psi \in H^k \cap \Sbb$, $\im \d_{[\psi]} \pt$ splits $\chh$ (see \cite[p766]{Zeidler1}), i.e. that there is a projection from $\chh$ onto $\im \d_{[\psi]} \pt$, where
 \begin{align*}
	 \im \d_{[\psi]} \pt \simeq \acs{ \ketbra{\vp}{\psi} + \ketbra{\psi}{\vp} \bigst \vp \in H^k \cap \acs{\psi}^{\perp}} \subset \chh.
 \end{align*}
	 First, $\im \d_{[\psi]} \pt$ is closed. We define the linear operator
\begin{align*}
	\gamma :
\begin{array}{ccc}
	\chh & \lra & \im \d_{[\psi]} \pt \\
	G & \longmapsto & P_{\acs{\psi}^{\perp}} G P_{\psi} + P_{\psi} G P_{\acs{\psi}^{\perp}}.
\end{array}
\end{align*}
	We decompose $H^k  = \vect \psi \overset{\perp}{\oplus} \acs{\psi}^{\perp}$, where the projections on each parts are continuous in $H^k \ra H^k$. We can represent an element $G \in \chh$ as 
 \begin{align*}
	  G = \alpha \proj{\psi} + \ketbra{\vp}{\psi} + \ketbra{\psi}{\vp} + M 
 \end{align*}
where $\alpha \in \R$, $\vp \in H^k \cap \acs{\psi}^{\perp}$ and $M \in \chh( \acs{\psi}^{\perp} )$. This is the division
 \begin{align*}
	 \chh = \im \gamma \oplus \im (1-\gamma),
 \end{align*}
	where $\oplus$ means that $\im \gamma \cap \im (1-\gamma) = \acs{0}$. We have $\gamma G = \ketbra{\vp}{\psi} + \ketbra{\psi}{\vp}$, thus $\gamma^2 = \gamma$.

	\textit{\bul $\im \gamma$ and $\im (1-\gamma)$ are closed in $\chh$.}
	For $\vp \in H^k \cap \acs{\psi}^{\perp}$, we define $G_{\vp} \df \ketbra{\psi}{\vp} + \ketbra{\vp}{\psi} \in \im \gamma$. Let $\vp_n \in H^k \cap \acs{\psi}^{\perp}$ be a sequence such that $G_{\vp_n} \ltend{\chh} G$ for some $G \in \chh$. We define $\vp \df G \psi$. We have $\vp - \vp_n = (G-G_{\vp_n}) \psi$ and then $\nor{\vp-\vp_n}{H^k} \le \nor{G-G_{\vp_n}}{\chh} \nor{\psi}{H^{-k}}$ and $\vp_n \ltend{H^k} \vp$. We have $G-G_{\vp_n} = G_{\vp-\vp_n}$ so $\nor{G_{\vp}-G_{\vp_n}}{\chh} \le 2 \nor{\psi}{H^k} \nor{\vp-\vp_n}{H^k}$ and $G_{\vp_n} \ltend{\chh} G_{\vp}$, so $G = G_{\vp}$. We conclude that $\im \gamma$ is closed in $\chh$.

	For $\alpha \in \R$ and $M \in \chh(\acs{\psi}^{\perp})$, we define $G_{\alpha,M} \df \alpha P_{\psi} + M \in \im (1-\gamma)$. Let $(\alpha_n,M_n) \in \R \times \chh(\acs{\psi}^{\perp})$ be such that $G_{(\alpha_n,M_n)} \ltend{\chh} G$ for some $G \in \chh$. We define $\alpha \df \ps{\psi, G \psi}$ and $M \df G - \alpha P_{\psi}$. We have
 \begin{align*}
	 & \alpha - \alpha_n = \ps{\psi,(G-G_{(\alpha_n,M_n)})\psi} \\
	 & \hs\hs\hs\hs = \ps{(-\Delta+1)^{-\f{k}{2}} \psi, \delk(G-G_{(\alpha_n,M_n)}) \delk (-\Delta+1)^{-\f{k}{2}} \psi}
 \end{align*}
so $\ab{\alpha-\alpha_n} \le \nor{\psi}{H^{-k}}^2 \nor{G-G_{(\alpha_n,M_n)}}{\chh}$ and $\alpha_n \ra \alpha$. Moreover, $M_n-M = G_n-G+(\alpha_n-\alpha) P_{\psi}$ so 
 \begin{align*}
\nor{M_n-M}{\chh} \le \nor{G_{(\alpha_n,M_n)}-G}{\chh} + \ab{\alpha_n-\alpha} \nor{\psi}{H^k}^2
 \end{align*}
and $M_n \ltend{\chh} M$. Eventually, 
 \begin{align*}
\nor{G_{(\alpha_n,M_n)} -G_{(\alpha,M)}}{\chh} \le \nor{M_n-M}{\chh} + \ab{\alpha_n-\alpha} \nor{\psi}{H^k}^2,
 \end{align*}
and $G_{(\alpha_n,M_n)} \ltend{\chh} G_{(\alpha,M)}$ so $G = G_{(\alpha,M)} \in \im (1-\gamma)$. We conclude that $\im (1-\gamma)$ is closed.

\textit{\bul $\gamma$ is continuous.}
Let $\cG_{\gamma} \df \acs{ (G, \gamma G) \st G \in \chh}$ be the graph of $\gamma$. Let $(G_n)_{n \in \N} \in \chh^{\N}$ be a sequence such that $G_n \ltend{\chh} G$ and $\gamma G_n \ltend{\chh} F$ for some $G,F \in \chh$. $\im \gamma$ is closed so $F \in \im \gamma$ and $\gamma F = F$. Also, $G_n-\gamma G_n = (1-\gamma) G_n \ltend{\chh} G-F$, but since $\im (1-\gamma)$ is closed, then $G-F \in \im (1-\gamma)$ so $0 = \gamma(G-F) = \gamma G - F$ and $\gamma G = F$. This proves that $\cG_{\gamma}$ is closed in $\chh$, and thus by the closed graph theorem, $\gamma$ is continuous.

\textit{\bul Conclusion.}
$\gamma$ is thus a projector and $\im \d_{[\psi]} \pt$ splits $\chh$. We conclude that $\pt$ is an embedding. \\

\textit{\bul $\pt\iv$ is $\cC^{\ii}$.} We know that $\pt$ is a $\cC^1$-embedding, thus at any point $\p$ and working in local charts, $\d_{\p} \pt$ is invertible and its image splits, so we can apply the inverse function theorem. We refer to \cite[Theorem 73.E]{Zeidler4} and \cite[Section I.5]{Lang72}. All the degrees of regularity of $\pt$ are passed on its inverse \cite[Proposition 5.3]{Lang72}.
\end{proof}
The proof of Corollary \ref{cc} consists in applying \cite[Theorem 73.E]{Zeidler4}, and the fact that the topologies are equivalent by \eqref{tops}. Since $\pt$ and $\pt\iv$ are $\cC^1$, then $\d_{\pt(\p)} \pt\iv = \pa{\d_{\p} \pt}\iv$ by the chain rule. Also, $\pt$ is bi-Lipschitz for $k=0$ by \eqref{bili}.

\section{Proofs: the wavefunction-to-density map $\rop$}


In this section, we provide a basic property on the map $\p \mapsto \ro_{\p}$ from $H^1\ind{p}(\R^{dN})$ to $H^1(\R^d)$. We define the map from a wavefunction to its one-body density,
\begin{align*}
\rop :
\begin{array}{ccl}
H^k\ind{p}   & \lra & W^{k,1}(\R^d) \cap \acs{ \int \cdot = N} \\
\end{array}
\end{align*}
by $\rop(\C \psi) \df \ro_{\psi}$, and we also use the notation $\ro_{\p} \df \rop(\p)$.
Its differential has to be defined in local charts, by
\begin{align*}
\d \rop :
\begin{array}{ccc}
H^k\ind{p}  & \lra & \cB \pa{ H^k\cap  \acs{\psi}^{\perp}, W^{k,1} \cap \acs{\int \cdot = 0} }\\
	\seg{\psi} & \longmapsto &  \d_{\seg{\psi}} \rop = 2N\re \int_{\R^{d(N-1)}} \ov{\psi} \hs \cdot, \\
\end{array}
\end{align*}
which depends on the point of $H^k$ at which we look, contrarily to $\rop$. The choice of $\psi$ is the choice of a relative phase in the corresponding chart. This section consists in proving that it is smooth as claimed in the next lemma.

\begin{lemma}[Smoothness of $\rop$]\label{smoothrop}
For any $k \in \N$, $\rop$ is $\cC^{\ii}$. For any $\p,\Phi \in H^k\ind{p}$, we have
\begin{align}\label{estiro}
\nor{ \ro_{\p} - \ro_{\Phi} }{W^{k,1}} \le c_{k,d} \bpa{ \nor{\p}{H^{k}\ind{p}} + \nor{\Phi}{H^{k}\ind{p}}}  \D_k\pa{\p,\Phi},
\end{align}
where $c_{k,d}$ is a constant depending only on $k$ and $d$. The map $\rop$ is nowhere injective and not proper. 
\end{lemma}

The map $\vp \mapsto \int \psi \vp$, from $L^2$ to $L^2$ for instance, is compact. For this reason we believe that $\d_{\p} \rop$ is a source of compactness in $\d_v \ro = \d_{\p(v)} \rop \rond \d_v \p$, which is itself a source of ill-posedness in the inverse Kohn-Sham problem.

\begin{proof}[Proof of Lemma \ref{smoothrop}]\tx{ }

	\bul \textit{Not proper.} We take $\ro \in \cC^{\ii}(\Omega,\R_+)$ such that $\int \ro = N$ and show that $\rop\iv \pa{\acs{\ro}}$ is not compact. Indeed, considering the Harriman-Lieb representation $\p_k$ of $\ro$, having an orbital with $(k,0,0)$ momentum \cite{Harriman81} \cite[proof of Theorem 1.2]{Lieb83b}, it satisfies $\ro_{\p_k} = \ro$ but $\nor{\p_k}{H^1} \ra + \ii$.

	\bul \textit{Continuity.} Let $Y \df \acs{ (x_2,\dots,x_N) \in \Omega^{N-1}}$. We have
	 \begin{align*}
		 \ab{\ro_{\psi} - \ro_{\vp}} &  \le \int_Y \ab{ \ab{\psi} - \ab{\vp} } \pa{ \ab{\psi} + \ab{\vp} } \d Y \le \int_Y \ab{\psi - \vp } \bpa{ \ab{\psi} + \ab{\vp} } \d Y,
	 \end{align*}
	 thus by integrating in the last variable we obtain
	 \begin{align*}
		 \nor{\sqrt{\ro_{\psi}} - \sqrt{\ro_{\vp}} }{L^2}^2 \le \nor{\ro_{\psi} - \ro_{\vp}}{L_1} &  \le \nor{\psi-\vp}{L^2} \sqrt{ \int \bpa{ \ab{\psi} + \ab{\vp} }^2} \le 2 \nor{\psi-\vp}{L^2}.
	 \end{align*}
	As for the derivatives,
 \begin{align*}
	 \na (\ro_{\psi} - \ro_{\vp}) = 2\re \int_Y \ov{\psi} \na \pa{ \psi-\vp} + 2\re \int_Y \ov{\psi-\vp} \na \vp,
 \end{align*}
	so
 \begin{align*}
	 \nor{\na\pa{\ro_{\psi} - \ro_{\vp}}}{L^1} & \le 2 \nor{\na \pa{\psi-\vp}}{L^2} + 2\nor{\na \vp}{L^2} \nor{\psi - \vp}{L^2} \\
	 & \le 2 \bpa{1 + \nord{\na \vp}} \nord{\del^{\ud} \pa{\psi-\vp}}.
 \end{align*}
	For the double derivatives, we can show that
 \begin{multline*}
	 \nor{\Delta\pa{\ro_{\psi} - \ro_{\vp}}}{L^1} \le 2 \pa{ \nor{\na \psi}{L^2} + \nor{\na \vp}{L^2}} \nor{\na \pa{\psi-\vp}}{L^2}\\
	 + 2 \nor{\psi-\vp}{L^2} \nor{\Delta \psi}{L^2} + 2 \nor{\vp}{L^2} \nor{\Delta \pa{ \psi-\vp}}{L^2},
 \end{multline*}
	and more generally for any $k\in \N$, there is a constant $c_{k,d}$, depending only on $k$ and on the dimension, such that
 \begin{align*}
	 \nor{\ro_{\psi}-\ro_{\vp}}{W^{k,1}} & \le c_{k,d} \sum_{i=1}^k \bpa{ \nor{\psi}{H^{k-i}} + \nor{\vp}{H^{k-i}}} \nor{\psi-\vp}{H^i}\\
	 & \le c_{k,d} \bpa{ \nor{\psi}{H^{k}} + \nor{\vp}{H^{k}}} \nor{\psi-\vp}{H^k}.
 \end{align*}
By changing the global phase for $\vp \ra e^{i\theta} \vp$, this leads to \eqref{estiro}.

	\bul \textit{Differentiability.} We see $H^k\ind{p}$ as a smooth manifold, with charts $c_{\psi}$ for $\psi \in H^k$ as considered in the proof of Proposition \ref{lka}. The description of $\rop$ is then done in those charts. We denote by $\ov{\ro}$ the map $H^k \ra W^{k,1}, \psi \mapsto \ro_{\psi}$. For $\vp \in H^k$ close to $\psi \in H^k \cap \Sbb$, we can represent $\rop$ by $\ov{\ro} \rond c_{\psi}\iv (\vp) = \ov{\ro} \pa{ \pi(\psi + \vp)} = \ro_{\psi + \vp}$, hence $\rop$ is smooth. 
	We have
 \begin{align*}
	 \ro_{\psi + \vp} - \ro_{\psi} - 2 N \re \int_{Y} \ov{\psi} \vp = \int_{Y} \ab{\vp}^2,
 \end{align*}
 with $\nor{  \int_{Y} \ov{\psi} \vp }{W^{k,1}} \le c \nor{\psi}{H^k} \nor{\vp}{H^k}$, therefore $2N\re \int_{Y} \ov{\psi} \cdot$ is bounded. Eventually,
 \begin{align*}
	 \nor{ \int_{Y} \ab{\vp}^2  }{W^{k,1}} \le c_{d,k} \nor{\vp}{L^2}\nor{\vp}{H^k}.
 \end{align*}
	Hence $\ov{\ro}$ has differential $\d_{\psi} \ov{\ro} = 2 N \re \int_Y \ov{\psi} \cdot$, which is a representative of $\d_{[\psi]} \rop$ in the chart $(U_{\psi},c_{\psi})$.
\end{proof}

%
%
%
The conclusions of this section hold for other potential-to-ground state density maps, for instance for the current map $\p \mapsto j_{\p}$ etc.

\section{Proofs: maps from potentials to ground state quantities}\label{mapsproofs} 
In this section, we prove the corresponding results of Theorem \ref{proposs}, but for the map $v \mapsto \proj{\p\ex{k}(v)}$ , and then transport them to the map $v \mapsto \p\ex{k}(v)$ using that $\pt$ is an embedding.

\subsection{The restriction} 

For any operator $A$ of $L^2(\Omega)$, we define $\tilde{A}_{\perp} \df (1-P_{\p}) \restr{A}{\acs{\p}^{\perp}}$ as an operator of $\acs{\p}^{\perp}$, and $A_{\perp} \df A (1-P_{\p})$ as an operator of $L^2$. We need to work in $\acs{\p}^{\perp}$ because it corresponds to the tangent space of $\Sbb$ at $\p$. We split $L^2 = \vect \p \oplus \acs{\p}^{\perp}$. Let us write $H$ and $E$ instead of $\hn(v)$ and $\exc{k}(v)$, $P$ is the orthogonal projection on $\Ker (H-E)$ and $P_{\perp} \df 1-P$.

\begin{proposition}[Properties of $H_{\perp}$]\label{equi} Let $p$ be as in \eqref{expot}, and let $v \in \spn{k}$, $H \df \hn(v)$, $E \df \exc{k}(v)$ and $\lambda \in \R$.

\begin{enumerate}[label=(\roman*),leftmargin=*]
	\item\label{ba1} As an operator on $\acs{\p}^{\perp}$, $\tilde{H}_{\perp}$ is self-adjoint on $D\bpa{H} \cap \acs{\p}^{\perp}$. On $\acs{\p}^{\perp}$, $\widetilde{(H-\lambda)}_{\perp}^{-1/2} = (H-\lambda)^{-1/2} $ for $\lambda \not \in \sigma(H)$. Moreover $\sigma(\tilde{H}_{\perp}) = \sigma(H) \backslash \acs{E}$ and $\sigma\ind{ess}(\tilde{H}_{\perp}) = \sigma\ind{ess}(H)$. 
\item\label{ba2} 
	The operator $(-\Delta +1)^{\ud}(H-E)\iv_{\perp}(-\Delta +1)^{\ud}$ is bounded.
\end{enumerate}
\end{proposition}

\begin{proof} The first part \ref{ba1} is well-known and follows from the spectral calculus \cite{ReeSim1,Lewin18}. Since $H$ and $(H-E)\iv_{\perp}$ commute, we have for $\lambda > 0$,
 \begin{align*}
	 & \bhs \norop{\del^{\ud} (H-E)\iv_{\perp}   \del^{\ud}} \\
	 & = \norop{\del^{\ud} \f{H-E+\lambda}{(H-E+\lambda)^{\ud}}  (H-E)\iv_{\perp}  (H-E+\lambda)^{-\ud} \del^{\ud}} \\
	 &  \le \norop{\del^{\ud}(H-E+\lambda)^{-\ud}}^2 \norop{(H-E+\lambda)(H-E)\iv_{\perp} } \\
	 &  = \norop{\del^{\ud}(H-E+\lambda)^{-\ud}}^2 \norop{\lambda (H-E)\iv_{\perp} + P_{\perp}} \\
	 &  \le \pa{\lambda \dist \pa{E,\sigma(H) \backslash \acs{E}}\iv + 1} \norop{\del^{\ud}(H-E+\lambda)^{-\ud}}^2.
 \end{align*}
 By decomposing
 \begin{align*}
	 & \del^{\ud}(H-E+\lambda)\iv\del^{\ud} \\
	 & \bhs = \pa{(H-E+\lambda)^{-\ud}\del^{\ud}}^* \pa{(H-E+\lambda)^{-\ud}\del^{\ud}},
 \end{align*}
 we obtain
 \begin{multline*}
	 \norop{\del^{\ud}(H-E+\lambda)^{-\ud}}^2 \\
	 =   \norop{\del^{\ud}(H-E+\lambda)\iv\del^{\ud}},
 \end{multline*}
	which is finite by Lemma \ref{continj}.
\end{proof}


\subsection{The potential-to-density matrix map $\gam$}

We define the map
\begin{align*}
\gam :
\begin{array}{ccl}
\spn{k} & \lra & \im \pt \\
	v & \longmapsto & P_{\p\ex{k}(v)} = \ketbra{\p\ex{k}(v)}{\p\ex{k}(v)} \\
\end{array}
\end{align*}
giving the ground state density matrix. Its differential is
\begin{align*}
\d \gam : 
\begin{array}{ccc}
	\spn{k} & \lra & \cB \pa{ \spa , \sch_{1,1}} \\
v & \longmapsto & \pa{u \mapsto \pa{\d_v \gam} u} \\
\end{array}
\end{align*}
and we recall that, if $\psi(v)$ is a representative of $\p\ex{k}(v)$, then
 \begin{align}\label{tn}
	  \pa{\d_v \gam}u & \in \pa{\d_{\p(v)} \pt} \pa{\Td_{\p(v)} H^k\ind{p}} \\
	 & \bhs \bhs =  \acs{ \ketbra{\psi(v)}{\phi} + \ketbra{\phi}{\psi(v)} \bigst \phi \in H^k \cap \acs{\psi(v)}^{\perp} }. \nonumber
 \end{align}
The differential of $\gam$ takes its values in the tangent space of $\im \pt$, corresponding to the tangent space of $H^k\ind{p}$. On this last space, the relevant operator acting on tangent vectors is $\widetilde{\hn(v)}_{\perp}$. For simplicity, we will use the notations
 \begin{align*}
	 \tx{$\ssum$ } v_i \df \sum_{i=1}^N v(x_i), \bhs\tx{$\ssumd$ }  w_{ij} \df \sum_{1 \le i \sle j \le N} w(x_i-x_j). 
 \end{align*}

We define the contour $\cC \df \acs{z \in \C \st \ab{z-\exc{k}(v)}=\eta(v)}$ where $\eta(v) \df \dist \bpa{ \exc{k}(v),\sigma(\hn(v))\backslash \acs{\exc{k}(v)}}/2$.

\begin{theorem}[Properties of $\gam$]\label{propgam} Let $p$ be as in \eqref{expot} and $\spa = \spp$. The potential-to-ground state density matrix map $\gam$ is $\cC^{\ii}$. At some point $v \in \spn{k}$, its differential is
\begin{align}\label{gameq}
	& \pa{\d_v \gam} u  =\inv{2\pi i} \oint_{\cC}\d z \bpa{z-\hn(v)}\iv \vv{u} \bpa{z-\hn(v)}\iv \\
& \hs\hs\hs = \bpa{\exc{k}(v)-\hn(v)}\iv_{\perp}  \vv{u} \gam(v) + \gam(v) \vv{u}  \bpa{\exc{k}(v)-\hn(v)}\iv_{\perp}. 
\end{align}
	Also, for any $v \in \spn{k}$, $\tr \bpa{(\d_v \gam) u} = 0$ for any $u \in \spp$. The differential $\d_v \gam$ is compact from $\spp$ to $\sch_{1,1}$ and not surjective. If $k=0$ and $p>\max(2d/3,2)$, then $\gam$ and $\d_v \gam$ are injective.
\end{theorem}

\begin{proof} \tx{ }

	\textit{\bul Continuity.} Let $v,u\in \spn{k}$ be such that $\nor{v-u}{L^p+L^{\ii}} \le \ep$ with $\ep$ so small that $\ab{\exc{k}(v)-\exc{k}(u)} \sle \eta(v)/8 \sle \eta(u)$. For all $z\in \C$ such that $\ab{z-\exc{k}(v)} = \eta(v)/2$, we have thus $\dist(z,\sigma\pa{\hn(u)}) \ge \eta(v)/8$. We use the resolvent formula and integrate over a contour $\cC$ located around $\exc{k}(v)$ and $\exc{k}(u)$, we have
	\begin{multline*}
\gam(v) - \gam(u) = \int_{\cC} (-\Delta+1)^{-\ud} C(z) (-\Delta+1)^{-\ud} \vv{(v-u)} (-\Delta+1)^{-\ud} \\
\times D(z) (-\Delta+1)^{-\ud},
\end{multline*}
where 
 \begin{align*}
 C(z) \df \del^{\ud} (z-\hn(v))\iv \del^{\ud} \\
 D(z) \df \del^{\ud} (z-\hn(u))\iv \del^{\ud}
 \end{align*}
are bounded uniformely in $z$, as justified by Lemma \ref{continj}. We estimate
 \begin{align*}
	 \nor{\gam(v) - \gam(u)}{\sch_{\ii,1}} & \le c \norop{(-\Delta+1)^{-\ud} \vv{(v-u)} (-\Delta+1)^{-\ud}} \\
	 & \le c N \nor{v-u }{\spp},
 \end{align*}
	where we used Lemma \ref{continj}. Finally, we saw in Corollary \ref{cc} that on $\im \gam$, the norms $\sch_{1,\ii}$ and $\sch_{1,1}$ are equivalent.

\textit{\bul Differentiability.} Let $v \in \spn{k}$ and $u \in \spp$ be small enough so that $v+u \in \spn{k}$. By the resolvent formula, we have
\begin{align*}
& \gam(v+u)- \gam(v) - \inv{2\pi i} \oint_{\cC } \pa{z-\hn(v)}\iv \vv{u} \pa{z-\hn(v)}\iv \\
& \bhs = \inv{2\pi i} \oint_{\cC} \pa{z-\hn(v+u)}\iv \cro{\vv{u} \pa{z-\hn(v)}\iv}^2 \\
	& \bhs = \inv{2\pi i} \oint_{\cC}\del^{-\ud} G(z) \del^{-\ud} \vv{u} \del^{-\ud} C(z)\\
	& \bhs\bhs\bhs \times \del^{-\ud} \vv{u} \del^{-\ud} C(z)\del^{-\ud},
\end{align*}
	where the operator $G(z) \df \del^{\ud}\bpa{z-\hn(v+u)}\iv \del^{\ud}$ is uniformly bounded in $z$. Therefore
\begin{align*}
	& \nor{\gam(v+u)- \gam(v) - \inv{2\pi i} \oint_{\cC } \bpa{z-\hn(v)}\iv \vv{u} \bpa{z-\hn(v)}\iv}{\sch_{\ii,1}} \\
	& \bhs\bhs\bhs\bhs\bhs\bhs\bhs\bhs   \le c \nor{u}{\spp}^2,
\end{align*}
where $c$ is independent of $u$. Hence $\d_v \gam$ exists and is given by the first equality in \eqref{gameq}.

\textit{\bul Formula \eqref{gameq}.} We denote by $\Lambda$ the domain delimited by $\cC$. First, because the only singularity inside $\Lambda$ is on $z = \exc{k}(v)$, we have by the spectral theorem
\begin{align*}
	0&  = \oint \bpa{z-\hn(v)}\iv \gam(v) \vv{u}  (z-\hn(v))\iv \gam(v) \\
	& = \oint \bpa{z-\hn(v)}\iv \gamp \vv{u}  (z-\hn(v))\iv \gamp.
\end{align*}
	Moreover, since $\Lambda$ and $\sigma\bpa{\widetilde{\hn(v)}_{\perp}}$ are disjoint, the spectral theorem implies
\begin{align*}
	& \inv{2 \pi i} \oint_{\cC}  (z-\hn(v))\iv  \gamp\vv{u}  (z-\hn(v))\iv \gam(v)\\
	& \bhs\bhs = \inv{2 \pi i} (\exc{k}(v)-\hn(v))\iv_{\perp} \vv{u}  \pa{ \oint_{\cC} (z-\hn(v))\iv} \gam(v) \\
	& \bhs\bhs = (\exc{k}(v)-\hn(v))\iv_{\perp}  \vv{u} \gam(v).
\end{align*}
Similarly,
\begin{align*}
	& \inv{2 \pi i} \oint_{\cC} (z-\hn(v))\iv \gam(v) \vv{u} (z-\hn(v))\iv \gamp \\
	&\bhs\bhs\bhs\bhs =  \gam(v) \vv{u}  (\exc{k}(v)-\hn(v))\iv_{\perp}.
\end{align*}

	\textit{\bul Regularity of the differential.} The following expressions are well-known \cite{Kato}. Let $v,h \in \spn{k}$ be potentials, close enough so that we can find a common relevant integration contour $\cC$, and $u \in \spp$ an element of the tangent spaces. We have
\begin{align*}
		& \bhs \pa{ \d_v \gam - \d_{h} \gam} u \\ 
		& = \inv{2 \pi i} \oint_{\cC}\hspace{-0.1cm}  \bigg( \hspace{-0.1cm}  \pa{z-\hn(v)}\iv \vv{u} \pa{z-\hn(v)}\iv \ssum (v-h)_i \pa{z-\hn(h)}\iv \\
		& \hs\hs\hs\hs\hs+ \pa{z-\hn(v)}\iv \bpa{\ssum (v-h)} \pa{z-\hn(h)}\iv \vv{u} \pa{z-\hn(h)}\iv \bigg)
\end{align*}
therefore $\nor{\pa{ \d_v \gam - \d_{h} \gam} u}{\sch_{\ii,1}} \le  c \nor{v-h}{\spp}  \nor{u}{\spp}$ and 
 \begin{align*}
\nor{ \d_v \gam - \d_{h} \gam}{\spp \ra \sch_{\ii,1}} \le c \nor{v-h}{\spp},
\end{align*}
and thus $v \mapsto \d_v \gam$ is locally Lipschitz.

By similar methods, we can show that $\gam$ is infinitely differentiable and that for any $m \in \N$, the $m^{\tx{th}}$ derivative is given by
\begin{align*}
	(\d_{v}^m \gam) (v_1,\dots,v_m) = \inv{2\pi i} \oint_{\cC}\bpa{z-\hn(v)}\iv \prod_{\ell=1}^m  \pa{\bpa{\ssum (v_{\ell})_i} \bpa{z-\hn(v)}\iv},
\end{align*}
with $\nor{(\d_{v}^m \gam) (v_1,\dots,v_m)}{\sch_{\ii,1}} \le c \prod_{ \ell =1}^m \nor{v_{\ell}}{\spp}$.

\textit{\bul $\tr (\d_v \gam) u = 0$.} This is because by definition the differential takes its values in the tangent space of the image space, but this can also be verified analytically.

\textit{\bul Injectivity of the differential.} Let $v\in \spn{k}$ and $u \in \spp$ be such that $\pa{\d_v \gam}u = 0$. We consider the representation \eqref{gameq} and let $\hn(v)-\exc{k}(v)$ act on the left, this yields $(1-\gam(v)) \vv{u} \gam(v) = 0$, that is $\vv{u} \p\ex{k}(v) = \p\ex{k}(v) \int u \ro_{\p\ex{k}(v)}$. By unique continuation (see \cite{Garrigue18,Garrigue19}), the nodal set of $\p\ex{k}(v)$ has zero measure, hence $\ssum u = \int u \ro_{\p\ex{k}(v)}$ and by integrating on $\seg{0,1}^{d(N-1)}$ we can conclude that $u$ is constant.

\textit{\bul Compactness of the differential.} Let us first show a lemma. We recall that a sequence of operators $L_n$ of $L^2(\R^n)$, such that $\norop{L_n} \le c$, converges strongly to 0 if $\nord{L_n f} \ra 0$ for any $f \in L^2(\R^n)$, and converges weakly to 0 if $\ps{g, L_n f} \ra 0$ for any $f, g \in L^2(\R^n)$.
\begin{lemma}\label{lemap}\tx{ }

	$(i)$ Let $L_n$ be a sequence of operators such that $\norop{L_n} \le c$, $L_n \wra 0$ weakly, and let $A$ and $B$ be two compact operators, then $\norop{A L_n B} \ra 0$. \smallskip

	$(ii)$ Let $L_n \df \del^{-\ud} u_n \del^{-\ud}$ with $u_n \wra 0$ in $\spp$ and $p$ as in \eqref{expot}. Then $\norop{L_n} \le c$ and $L_n \ra 0$ strongly.  \smallskip

	$(iii)$ Let $L_n$ be a sequence of operators of $L^2(\R^n)$ such that $\norop{L_n} \le c$ and $L_n \ra 0$ strongly, and let $A \in \sch_1(\R^n)$ be a self-adjoint trace-class operator. Then $\tr \bpa{L_n A L_n} \ra 0$.
\end{lemma}
\begin{proof}
$(i)$ The set of finite rank operators is dense in the set of compact operators so by an \apo{$\ep/2$} argument, we can assume that $A$ and $B$ have finite rank. So let us write them $A = \sum_{i=1}^m \ketbra{f_i}{g_i}$ and $B = \sum_{i=1}^m \ketbra{h_i}{u_i}$. We have
 \begin{align*}
	 \norop{A L_n B} &= \norop{\sum_{1 \le i,j \le m} \ps{g_i,L_n h_i} \ketbra{f_i}{u_i} } \\
	 & \le \sum_{1 \le i,j \le m} \ab{\ps{g_i,L_n h_i} } \nord{f_i} \nord{u_i},
 \end{align*}
and we conclude by letting $n \ra +\ii$, where $\ps{g_i,L_n h_i} \ra 0$. \\

	$(ii)$ As we showed in Lemma \ref{continj}, $\norop{L_n} \le c \nor{u_n}{\spp}$, with $c$ independent of $n$, hence $\norop{L_n}$ is bounded. By density of $\cC^{\ii}\ind{c}(\R^d)$ in $L^2(\R^d)$, we hence only need to show that $\nord{L_n f} \ra 0$ for any $f \in \cC^{\ii}\ind{c}(\R^d)$. So let $f \in \cC^{\ii}\ind{c}(\R^d)$, and take a function $\chi$ of $\R_+$, equal to 1 on $[0,1]$, vanishing on $[2,+\ii)$ and smooth and decreasing on $[1,2]$, and define the localization function $\chi_r(x) \df \chi \bpa{\ab{x}/r}$ on $\R^d$. We take $r$ large enough so that $\supp f \subset B_r$. We have
	 \begin{align*}
 \nord{L_n \chi_r^2 f} =  \nord{L_n \chi_r^2 \del^{-\ud} \del^{\ud} f}  
	 \end{align*}
	 and since $\del^{\ud} f \in L^2(\R^d)$, we only need to show that $L_n \chi_r^2 \del^{-\ud}$ converges strongly to 0. We will in fact prove that 
 \begin{align*}
	 \mylim{r \ra + \ii} \mylim{n \ra +\ii} \norop{L_n \chi_r^2 \del^{-\ud}} = 0.
 \end{align*}
	 Let us consider the decomposition
 \begin{align*}
	 & L_n \chi_r^2\del^{-\ud}  \\
	 & \hs\hs\hs\hs = \del^{-\ud} u_n \del^{-\ud} \chi_r \seg{\chi_r,\del^{-\ud}}  \\
	 & \bhs + \del^{-\ud} u_n \seg{\del^{-\ud},\chi_r} \del^{-\ud} \chi_r   \\
	 & \bhs + \seg{\del^{-\ud},\chi_r} u_n \del\iv \chi_r \\
	 & \bhs + \chi_r \del^{-\ud} u_n \del\iv \chi_r.
 \end{align*}
	We also have
 \begin{align*}
	 \seg{\del^{-\ud}, \chi_r} = -\del^{-\ud} \seg{\del^{\ud},\chi_r} \del^{-\ud},
 \end{align*}
	For $r$ large enough, we have $\norop{\seg{\del^{\ud}, \chi_r}} \le c / r$ for some constant $c$ independent of $r$ \cite[Lemma 1]{HaiLewSer05b}, hence
	\begin{align}\label{cui}
		\norop{ L_n \chi_r^2 \del^{-\ud}  - \chi_r \del^{-\ud} u_n \del\iv \chi_r } \le c /r,
 \end{align}
	where $c$ is independent of $n$ and $r$. Now we decompose
 \begin{align*}
	 & \chi_r \del^{-\ud} u_n \del\iv \chi_r \\
	 & \bhs = \chi_r \del^{-\f{\ep}{2}} \del^{-\f{1-\ep}{2}} \ab{u_n}^{\f{1-\ep}{2}} \\
	 & \bhs \bhs  \times \sgn(u_n) \ab{u_n}^{\f{1+\ep}{2}} \del^{-\f{1+\ep}{2}}\del^{-\f{1-\ep}{2}}  \chi_r.
 \end{align*}
	For $d \ge 3$, $p > d/2$, and for $\ep$ small enough we can still use the HLS inequality to prove that 
 \begin{align*}
  \norop{ \ab{u_n}^{\f{1\pm \ep}{2}} (-\Delta)^{-\f{1\pm \ep}{2}}} \le c_d \nor{u_n}{L^p}^{\f{1\pm \ep}{2}}.
 \end{align*}
For $d = 2$, we can still use the Kato-Seiler-Simon inequality, where we have to take $0 \sle \ep \sle \min(p-1,1)$, so $2p/(1\pm \ep) > 2$ and
 \begin{align*}
	 & \norop{\ab{u_n}^{\f{1\pm \ep}{2}} (-\Delta+1)^{-\f{1\pm \ep}{2}} } \\
	 & \bhs \le (2\pi)^{-\f{d(1 \pm \ep)}{p}} \nor{ \ab{u_n}^{\f{1\pm \ep}{2}} }{L^{\f{2p}{1\pm \ep}}} \nor{\bpa{\ab{x}^2+1}^{-\f{1\pm \ep}{2}}}{L^{\f{2p}{1\pm \ep}}} \\
	 & \bhs = c_{d,p,\ep}\nor{\bpa{\ab{x}^2+1}\iv}{L^p}^{\f{1\pm \ep}{2}}\nor{u_n}{L^{p}}^{\f{1 \pm \ep}{2}}.
 \end{align*}
	For $d = 1$, $p=1$ we also use the same argument. Hence 
	\begin{align}\label{lml}
\del^{-\f{1-\ep}{2}} u_n \del^{-\f{1+\ep}{2}}
 \end{align}
is bounded uniformly in $n$. Let us take $h$ and $g$ in the Schwartz space $\cS^{\ii}(\R^d)$. We have 
 \begin{multline*}
	 \ps{h, \del^{-\f{1-\ep}{2}} u_n \del^{-\f{1+\ep}{2}} g} \\
	 = \int u_n \pa{ \del^{-\f{1-\ep}{2}} h} \pa{ \del^{-\f{1+\ep}{2}} g},
 \end{multline*}
	and by regularity of $h$ and $g$, $\pa{\del^{-\f{1-\ep}{2}}  h} \pa{\del^{-\f{1+\ep}{2}} g}$ is in $(L^1 \cap L^{p'})(\R^d)$ so the above expression converges to 0 when $n \ra +\ii$. By density of $\cC^{\ii}(\R^d)$ in $L^2(\R^d)$, this shows that \eqref{lml} converges weakly to 0 when $n \ra +\ii$. Finally, since $\chi_r \del^{-\f{\ep}{2}}$ and $\del^{-\f{1-\ep}{2}}\chi_r$ are compact, then $\chi_r \del^{-\ud} u_n \del^{-1} \chi_r \ra 0$ strongly by applying Lemma~\ref{lemap} $(i)$. Considering \eqref{cui} again, by choosing $r$ large and then $n$ large, we can make $\norop{L_n \chi_r^2 \del^{-\ud}}$ arbitrarily small. \\

$(iii)$ We consider the representation $A = \sum_{i=1}^{\ii} \lambda_i \proj{f_i}$ where $(f_i)_i$ is an orthonormal familly of $L^2(\R^d)$ and $\sum_{i=1}^{\ii} \ab{\lambda_i} \sle +\ii$. Take $m \in \N$, we have
 \begin{align*}
\ab{\tr \bpa{L_n A L_n}} = \ab{\sum_{i=1}^{\ii} \lambda_i \nord{L_n f_i}^2} \le \sum_{i=1}^m \ab{\lambda_i} \nord{L_n f_i}^2 + c \sum_{i \ge m+1} \ab{\lambda_i}.
 \end{align*}
Let $\ep > 0$. By choosing $m$ large enough, we can have $c\sum_{i \ge m+1} \ab{\lambda_i} \le \ep/2$ and then, since $\nord{L_n f_i} \ra 0$, we can also choose $n$ large enough so that $\sum_{i=1}^m \ab{\lambda_i} \nord{L_n f_i}^2 \le \ep/2$.
 \end{proof}

We denote by $\psi$ a representative of $\p\ex{k}(v)$. Let $u_n \in \spp$ be such that $u_n \wra 0$ in $\spp$. We have
 \begin{align*}
 & \nor{ \bpa{\hn(v)-\exc{k}(v)}\iv_{\perp} \vv{(u_n)} \proj{\psi}}{\sch_{\ii,1}} \\ 
 & \bhs\bhs \bhs  = \nor{\psi}{H^1} \nor{ \bpa{\hn(v)-\exc{k}(v)}\iv_{\perp} \vv{(u_n)} \psi}{H^1}.
 \end{align*}
We want to show that the following quantity converges to zero,
\begin{align*}
 & \nor{ \bpa{\hn(v)-\exc{k}(v)}\iv_{\perp} \vv{(u_n)} \psi}{H^1}\\
	& \bhs \le \norop{ \del^{\ud}\bpa{\hn(v)-\exc{k}(v)}_{\perp}^{-1}\del^{\ud}} \\
 & \bhs\bhs\times \nord{\del^{-\ud}  \vv{(u_n)} \psi} \\
	& \bhs \le c_v \sum_{i=1}^N \nord{\del^{-\ud} u_n(x_i) \psi} \\
	& \bhs \le c_v \sum_{i=1}^N \nord{\deli^{-\ud} u_n(x_i) \psi}.
\end{align*}
We define $L_n \df \del^{-\ud} u_n \del^{-\ud}$ and notice that
 \begin{align*}
	 \nord{\deli^{-\ud} u_n(x_i) \psi}^2 = \inv{N} \tr_{\R^d} \pa{ L_n \del^{\ud} \gamma_{\psi} \del^{\ud} L_n},
 \end{align*}
 where $\gamma_{\psi}$ is the one-particle density matrix and $\del^{\ud} \gamma_{\psi} \del^{\ud} \in \sch_1$. By Lemma~\ref{lemap} $(ii)$, the operator $L_n$ converges strongly to 0 as an operator of $L^2(\R^d)$. Finally we apply Lemma~\ref{lemap} $(iii)$ to deduce that 
 \begin{align*}
 \nord{\deli^{-\ud} u_n(x_i) \psi}^2 \ra 0.
 \end{align*}

By the open mapping theorem, an operator cannot be compact and surjective, hence $\d_v \gam$ is not surjective as an operator from $L^p+L^{\ii}$ to $\sch_{1,1}$.
\end{proof}

\subsection{The potential-to-eigenstate map $\p\ex{k}$ } 

 \begin{proof}[Proof of Theorem \ref{proposs}]\tx{ }
	
	 \bul The properties $(i)$ and $(ii)$ are deduced from the composition $\p\ex{k} = \pt\iv \rond \gam$ and from Theorem \ref{propgam}. We only have to prove the expression \eqref{peq}. We remark that
 \begin{align*}
	 \pa{\d_v \gam} u & = \ketbra{\pa{\exc{k}(v)-\hn(v)}\iv_{\perp}   \vv{u} \p\ex{k}(v)}{\p\ex{k}(v)} \\
	 & \bhs\bhs + \ketbra{\p\ex{k}(v)}{\p\ex{k}(v) \vv{u}  \pa{\exc{k}(v)-\hn(v)}\iv_{\perp}} \\
	 & = \pa{\d_{\p\ex{k}(v)} \pt}\pa{\pa{\exc{k}(v)-\hn(v)}\iv_{\perp}  \vv{u} \p\ex{k}(v) }.
 \end{align*}
Now since $\pt\iv$ is $\cC^1$, we have
 \begin{align*}
	 \pa{\d_v \p\ex{k}}u & = \d_{v} \pa{\pt\iv \rond  \gam} u = \pa{\d_{\p\ex{k}(v)} \pt}\iv \rond  \pa{\d_v \gam} u \\
	 & = \pa{\exc{k}(v)-\hn(v)}\iv_{\perp}  \vv{u} \p\ex{k}(v).
 \end{align*}
We have
 \begin{align*}
	 \nor{\pa{\d_v \p\ex{k}}u}{H^1} & \le \nor{\del^{-\ud} \pa{\exc{k}(v)-\hn(v)}\iv_{\perp}\del^{-\ud}}{L^2}  \\
	 & \bhs\bhs \times \nor{\del^{\ud} \sqrt{\ssum \ab{u}}}{L^2} \nor{\sqrt{\ssum \ab{u}} \p\ex{k}(v)}{L^2} \\
	 & \le c_{v} \pa{\nor{u}{\spp} \int \ab{u} \ro_{\p\ex{k}(v)}}^{\ud}.
 \end{align*}

	 \bul $(iii)$ Let $v_n \in \sppe$ be a sequence which converges to 0 weakly. Let $\p_n$ be an approximate minimizer of $\exc{0}(v+v_n)$, that is
	 \begin{align}\label{ldeu}
	 \cE_{v+v_n}(\p_n) \le \exc{0}(v+v_n) + \f{1}{n},
 \end{align}
	and $\ro_n \df \ro_{\p_n}$.

 \bul If $d \ge 3$ and $p > d/2$, then we know that 
 \begin{align*}
	 \ab{v_n} \le c_{d,s} \nor{v_n}{(L^p+L^{\ii})(\Omega)} \pa{(-\Delta)^{1-s} + 1}
 \end{align*}
in the sense of quadratic forms, for some $s > 0$ depending on $p$. But $\nor{v_n}{(L^p+L^{\ii})(\Omega)}$ is bounded in $n$, and for any $\ep>0$, we have $(-\Delta)^{1-s} \le (1-s) \ep (-\Delta) + s \ep^{-1+1/s}$. We thus proved that for any $\ep > 0$ there is some $c_{\ep} \in \R$ independent of $n$ such that $\ab{v_n} \le \ep(-\Delta) + c_{\ep}$ in the sense of forms in $\Omega$. In dimensions $d \in \acs{1,2}$, the same holds under our assumptions on $p$.

We deduce that for any $\ep >0$, we have 
 \begin{align}\label{lun}
(1-\ep) \int \ab{\na \p}^2 - c_{\ep} \le \cE_{v_n+v}(\p)
\end{align}
uniformly in $\p$ and in $n$, for some $c_{\ep} \ge 0$.

\bul Next we prove that $\int v_n \ro_n \ra 0$. First, take some wavefunction $\Phi \in \wedge^N H^1(\Omega)$, we have
 \begin{align*}
 \exc{0}(v+v_n) \le \cE_v(\Phi) + \int v_n \ro_{\Phi} \le \cE_v(\Phi) + c_{d,N} \nor{v_n}{L^p+L^{\ii}} \nor{\sqrt{\ro_{\Phi}}}{H^1},
 \end{align*}
and since $\nor{v_n}{L^p+L^{\ii}}$ is bounded, then $\exc{0}(v+v_n)$ as well. Using it with \eqref{ldeu} and \eqref{lun}, we deduce that $\p_n$ is bounded in $H^1(\Omega)$ and $\p_n \wra \p_{\ii}$ weakly in $H^1(\Omega)$ for some $\p_{\ii} \in H^1(\Omega)$ and up to a subsequence. We have $\int \ab{\na \sqrt{\ro_n}}^2 \le \int \ab{\na\p_n}^2$ so $\sqrt{\ro_n}$ is bounded in $H^1(\Omega)$, hence there is some $\chi \ge 0$ in $H^1(\Omega)$ such that $\sqrt{\ro_n} \wra \chi$ weakly in $H^1(\Omega)$ hence strongly in $L^2(\Omega)$ locally. We define $\ro_{\ii} \df \chi^2$. 
Let $\ep > 0$, and let us decompose $\int v_n \ro_n$ into
\begin{align}\label{innn}
& \ab{\int_{\Omega} v_n \ro_n}  \le  \ab{\int_{B_r \cap \Omega} v_n \pa{\ro_n - \ro_{\ii}}} + \ab{\int_{B_r \cap \Omega} v_n \ro_{\ii}}+\ab{\int_{\Omega \backslash B_r} v_n \ro_n} \nonumber \\
& \hs\hs\hs\hs \le  \ab{\int_{B_r\cap \Omega}v_n\pa{\sqrt{\ro_n}-\sqrt{\ro_{\ii}}}\pa{\sqrt{\ro_n}+\sqrt{\ro_{\ii}}}} + \ab{\int_{B_r\cap \Omega} v_n \ro_{\ii}} \\
	& \bhs\bhs\bhs\bhs\bhs\bhs + \nor{v_n}{(\spp)(\Omega \backslash B_r)} \mysup{n \in \N} \nor{\sqrt{\ro_n}}{H^1} \nonumber.
 \end{align}
Also, the sequence $\nor{v_n}{L^p+L^{\ii}}$ is bounded. We take $r$ large enough so that $\Lambda \subset B_r$, and recall that $v_n \indic_{\Omega \backslash B_r} \ra 0$ strongly. Then we take $n$ large enough so that the last term in \eqref{innn} is smaller than $\ep$, which is possible since $v_n \indic_{\Omega \backslash \Lambda} \ra 0$. We also take $n$ large enough so that the second term in \eqref{innn} is smaller than $\ep$. As for the first term, we will need that for any functions $f,g,h$ in the appropriate spaces,
 \begin{align*}
\ab{ \int f g h} & \le \nor{f}{L^{p + \delta}} \nor{g}{L^{\f{2d}{d-2} - \eta(\delta)}} \nor{h}{L^{\f{2d}{d-2}}} \le \nor{f}{L^{p + \delta}} \nor{g}{W^{1-\lambda,\f{2d}{d-2\lambda}-\xi(\delta,\lambda)}} \nor{h}{H^1}.
 \end{align*}
 where $\eta(\delta) \df \f{16\delta}{(d-2)(d-2+2\delta(1+2/d))}$, $\xi(\delta,\lambda) \df  \f{\eta(\delta)(d-2)}{(d-2\lambda)\pa{1+\f{1-\lambda}{d}\pa{\f{2d}{d-2}-\eta(\delta)}}}$, this holds for any $\delta \ge 0$ and any $\lambda \in ]0,1[$ small enough, and we used the H\"older and Gagliardo-Nirenberg inequalities.
We apply it to our decomposition, where $\sqrt{\ro_n} \ra \sqrt{\ro_{\ii}}$ strongly in $W^{1-\lambda,\f{2d}{d-2\lambda}-\xi(\delta,\lambda)}(B_r\cap \Omega)$ by the theorem of Rellich-Kondrachov, where $\lambda > 0$ and $\delta >0$ are close to zero. This term is smaller than $\ep$ for $n$ large enough, and we conclude that $\int v_n \ro_n \ra 0$ for the considered subsequence. 

 
 \bul Since $w \ge 0$, then $\exc{0}$ is weakly upper semi-continuous for instance by a similar proof as for \cite[Theorem 3.6]{Lieb83b}, thus $\limsup \exc{0}(v+v_n) \le \exc{0}(v)$. Moreover, endowed with the geometric topology defined in \cite[Definition 2.1]{Lewin11}, the set of states on the Fock space with number of particles less than $N$ is compact as shown in \cite[Lemma 2.2]{Lewin11}. Let $\Gamma_{\ii} = G_0 \oplus \cdots \oplus G_N$ with $\tr \Gamma_{\ii} =1$ be a geometric limit, up to a further subsequence, of $\proj{\p_n}$. Since
\begin{align*}
\cE_{v+v_n}\pa{ \p_n } = \cE_v\pa{ \p_n } + \int v_n \ro_n \le \exc{0}(v+v_n) +\f{1}{n},
\end{align*}
then by weak semi-continuity of $\cE_v$ under geometric convergence \cite[Lemma 2.4]{Lewin11}, we have
\begin{align}\label{decf}
	\liminf \exc{0}(v+v_n) & \ge \cE_v(\Gamma_{\ii}) = \sum_{m=0}^N \tr H^m(v)G_m  \ge \sum_{m=0}^N E\ex{0}_m(v) \tr G_m  \nonumber \\
& \ge \exc{0}(v).
\end{align}
We used the HVZ theorem in the form of \cite[Theorem 3.1]{Lewin11} to deduce that $\excn{0}{M}(v)$ is decreasing in $M$ in the last inequality. We thus have $\exc{0}(v+v_n) \ra \exc{0}(v) = \cE_v(\Gamma_{\ii})$. We did not use the assumption $v \in \spfn{0}{N}$ yet, and repeating the same argument, we can also deduce that $E\ex{0}_m(v+v_n) \ra E\ex{0}_m(v)$ for any $m \in \acs{1,\dots,N}$. 

Since $\excn{0}{N-1}(v+v_n) \ra \excn{0}{N-1}(v)$, then
 \begin{align*}
	 \sign(v+v_n) = \excn{0}{N-1}(v+v_n) \ra \excn{0}{N-1}(v) = \sign(v) > \exc{0}(v),
\end{align*}
where we also used that $v \in \spfn{0}{N}$ and $\excn{0}{N-1}(v) = \sign(v)$ by the HVZ theorem. Hence for $n$ large enough, we have $\exc{0}(v+v_n) \sle \sign(v+v_n)$ and $v+v_n \in \spfn{0}{N}$. 

Since $v \in \spfn{0}{N}$ and by the HVZ theorem, we have $\exc{0}(v) \sle \excn{0}{N-1}(v)$ so since $E_m\ex{0}(v)$ strictly decreases in $m$, we have $\exc{0}(v) \sle \excn{0}{m}(v)$ for any $m \in \acs{0,\dots, N-1}$. By considering \eqref{decf} again, we have $\sum_{m=0}^N \excn{0}{m}(v) \tr G_m = \exc{0}(v)$ and can deduce that $\tr G_N = 1$. This yields $\proj{\p_n}$ converges geometrically to $\Gamma_{\ii}$, with $\Gamma_{\ii}$ being an operator of $\Ker \bpa{ \hn(v) - \exc{0}(v)}$. By \cite[Lemma~2.1]{Lewin11}, we deduce that $\proj{\p_n} \wra \Gamma_{\ii}$ in $\sch_{1,0}$, and since 
 \begin{align*}
	 \nor{\proj{\p_n}}{\sch_{1,0}} = 1 = \nor{\Gamma_{\ii}}{\sch_{1,0}},
 \end{align*}
then $\proj{\p_n} \ra \Gamma_{\ii}$ in $\sch_{1,0}$. Since $\p_n \wra \p_{\ii}$ weakly in $H^1(\R^d)$, then we also have $\tr \Gamma \pa{ \proj{\p_n}-\proj{\p_{\ii}}} \ra 0$ for any $\Gamma \in \sch_{1,0}$, and by uniqueness of the limit of $\proj{\p_n}$, we can conclude that $\Gamma_{\ii} = \proj{\p_{\ii}}$. By \eqref{mmm} and Corollary \ref{cc}, $\proj{\p_n} \ra \proj{\p_{\ii}}$ in $\sch_{1,0}$ up to a subsequence implies $\p_n \ra \p_{\ii}$ in $L^2$ up to a subsequence. Moreover, $\p_n$ converges to $\p_{\ii}$ weakly in $H^1$, and since the norm associated to $\cE_{v}$ is equivalent to the $H^1$ one and $\cE_v(\p_n) \ra \cE_v(\p_{\ii})$, then $\p_n \ra \p_{\ii}$ strongly in $H^1$ up to a subsequence. Finally, the same reasoning can be applied to any subsequence of $\p_n$, and we can conclude that 
 \begin{align*}
	P_{\Ker \bpa{\hn(v)-\exc{0}(v)}^{\perp}} \p_n \ra 0  
 \end{align*}
 in $H^1$, for the whole sequence. By continuity of the map $\p \mapsto \ro$, we also have $\ro_{\ii} = \ro_{\p_{\ii}}$. \\

\bul $(iv)$ When $\Omega$ is bounded, Lemma \ref{comps} implies that $v \mapsto \p\ex{0}(v)$ and its differential are compact, and $(\p\ex{0})\iv$ is discontinuous.
\end{proof}

We turn to the proof of Proposition \ref{propru}.
\begin{proof}[Proof of Proposition \ref{propru}] Here we will write $\psi$ for $\psi\ex{k}(v)$ and $\psi_n$ for $\psi\ex{k}(v_n)$, and define $V_n \df v_n - \exc{k}(v_n)/N$ and $V \df v - \exc{k}(v)/N$. We have Schrödinger's equations
\begin{align*}
\bpa{\ssum (V_n)_i} \psi_n  = \bpa{\Delta - \ssumd w_{ij}} \psi_n, \bhs \bpa{\ssum V_i} \psi  = \bpa{\Delta - \ssumd w_{ij}} \psi.
\end{align*}
	Since $p > \max(2d/3,2)$, then $w$ and $V_n$ are infinitesimally bounded by $(-\Delta)$ in the sense of operators with uniform constants, and therefore
\begin{align*}
	\nor{\ssum \pa{V_n-V}_i \psi}{L^2} & = \nor{ \bpa{-\Delta + \ssumd w_{ij} + \ssum (V_n)_i} (\psi_n - \psi)}{L^2} \\
 & \le c_{d,N} \pa{c_{w,d} + \nor{V_n}{L^p+L^{\ii}}} \nor{\psi_n - \psi}{H^2}.
\end{align*}
Since $V_n$ is bounded in $L^p+L^{\ii}$, then $\nor{\ssum (V_n-V)_i \psi}{L^2} \ra 0$. 
We also deduce that $\ssum (V_n-V)_i \psi \ra 0$ a.e. in $\Omega^N$ up to a subsequence. By unique continuation \cite{Garrigue19}, the nodal set
\begin{align*}
S \df \acs{x \in \Omega^N \bigst \psi(x) = 0}
\end{align*}
has zero measure in $\Omega^N$ and we deduce that 
\begin{align*}
\sum_{i=1}^N V_n(x_i) \ltend{n \ra +\ii}  \sum_{i=1}^N V(x_i),
\end{align*}
a.e. in $\Omega^N$, up to a subsequence. We can deduce that $V_n \ra V$ a.e. up to a subsequence by using Lemma \ref{lemale} provided at the end of this proof. Since $V_n-V$ is bounded, then $V_n \wra V$ weakly in $\spp$. We have
 \begin{multline*}
	 \nor{\ssum \pa{V_n-V}_i \psi}{L^2}^2 \\
	 = \int_{\Omega} (V_n-V)^2\ro_{\psi}  + 2 \int_{\Omega^2}  (V_n-V)(x)(V_n-V)(y) \rod_{\psi}(x,y)\d x \d y,
 \end{multline*}
	where $\rod_{\psi}(x,y) \df N(N-1)/2 \int \ab{\psi}^2 (x,y,x_3, \dots, x_N) \d x_3 \cdots \d x_N$ is the pair density of $\psi$. Since $V_n \wra V$ weakly, then 
 \begin{align*}
	 \int_{\Omega^2}  (V_n-V)(x)(V_n-V)(y) \rod_{\psi}(x,y)\d x \d y  \ra 0 
 \end{align*}
	and we conclude that $\int_{\Omega} (V_n-V)^2\ro_{\psi} \ra 0$.
\end{proof}

\begin{lemma}\label{lemale}
Let $v_n \in L^1\loc(\R^d)$. Then $v_n \ltend{n \ra +\ii} 0$ a.e. in $\R^d$ if and only if $\ssum v_n(x_i) \ltend{n \ra +\ii} 0$ a.e. in $\R^{dN}$. 
\end{lemma}
\begin{proof}
Let $S \subset \R^d$ be the set of $x$'s such that $v_n(x) \ltend{n \ra +\ii} 0$. Then for $(x_1,\dots,x_N) \in S^N$, we have $\ssum v_n(x_i) \ltend{n \ra +\ii} 0$, and $S^N$ has full measure in $\R^{dN}$.

For the converse statement, we define 
\begin{align*}
L \df \acs{x \in \R^{d} \bigst v_n(x) + \sum_{i=1}^{N-1} v_n(x_i) \us{n \ra +\ii}{\lra} 0 \tx{ a.e. in } (x_1,\dots,x_{N-1}) \in \R^{d(N-1)}},
\end{align*}
and for any $x \in L$ we define
\begin{align*}
L_x \df \acs{ (x_1,\dots,x_{N-1}) \in \R^{d(N-1)} \bigst v_n(x) + \sum_{i=1}^{N-1} v_n(x_i) \us{n \ra +\ii}{\lra} 0}.
\end{align*}
By the theorem of Fubini, $L$ has full measure in $\R^d$ and $L_x$ has full measure in $\R^{d(N-1)}$. We also define
\begin{align*}
L' \df \acs{ (x_1,\dots,x_{N-1}) \in \R^{d(N-1)} \bigst v_n(y) + \sum_{i=1}^{N-1} v_n(x_i) \us{n \ra +\ii}{\lra} 0 \tx{ a.e. } y \in \R^{d} }.
\end{align*}
and for $(x_1,\dots,x_{N-1}) \in L'$,
\begin{align*}
L'_{(x_1,\dots,x_{N-1})} \df \acs{ y \in \R^{d} \bigst v_n(y) + \sum_{i=1}^{N-1} v_n(x_i) \us{n \ra +\ii}{\lra} 0}.
\end{align*}
$L'$ has full measure in $\R^{d(N-1)}$ and $L_{(x_1,\dots,x_{N-1})}$ has full measure in $\R^{d}$. Now let $y_1,\dots,y_{N-1} \in L$ be such that $(y_1,\dots,y_{N-1}) \in L'$, and take $y_1',\dots,y_{N-1}' \in L'_{(y_1,\dots,y_{N-1})} \cap L$. We have 
\begin{align*}
v_n(y'_i) + \sum_{k=1}^{N-1} v_n(y_k)  \us{n \ra +\ii}{\lra} 0
\end{align*}
for any $i \in \acs{1,\dots,N-1}$, therefore
\begin{align*}
\sum_{k=1}^{N-1} v_n(y'_k) + (N-1)\sum_{k=1}^{N-1} v_n(y_k)\us{n \ra +\ii}{\lra} 0.
\end{align*}
Now let $z \in L'_{(y_1,\dots,y_{N-1})} \cap L'_{(y'_1,\dots,y'_{N-1})}$, we have
\begin{align*}
v_n(z) + \sum_{k=1}^{N-1} v_n(y'_k) \us{n \ra +\ii}{\lra} 0,
\end{align*}
and therefore
 \begin{align*}
 \inv{N-1} v_n(z) - \sum_{k=1}^{N-1} v_n(y_k)\us{n \ra +\ii}{\lra} 0,
 \end{align*}
 Summing it with,
 \begin{align*}
v_n(z)  + \sum_{k=1}^{N-1} v_n(y_k) \us{n \ra +\ii}{\lra} 0,
 \end{align*}
which holds because $z \in L'_{(y_1,\dots,y_{N-1})}$, we obtain $v_n(z) \ltend{n \ra +\ii} 0$.
Since $L'_{(y_1,\dots,y_{N-1})} \cap L'_{(y'_1,\dots,y'_{N-1})}$ has full measure in $\R^d$, then $v_n(z) \ltend{n \ra +\ii} 0$ a.e. in $z \in \R^d$.
\end{proof}

\subsection{The potential-to-density map $\ro\ex{k}$}

Properties of $\p\ex{k}$ can be lifted to properties on $\ro$. Here we define $\ro\ex{k}(v) \df \ro_{\p\ex{k}(v)} = \rop \rond \p\ex{k}(v)$, so $\ro = \ro\ex{0}$. 
\begin{theorem}[Properties of $\ro\ex{k}$]\label{maa} Take $\spa = \spp$ with $p$ as in \eqref{expot}, and $w \in \spp$.
	
	\bul \textit{(i - Smoothness).} The map $\ro\ex{k}$ is $\cC^{\ii}$ from $\spn{k}$ to $W^{1,1}\cap \acs{\int \cdot = N}$, and when $p > \max(2d/3,2)$ and $\spa = (\spp)/\sim$, $\ro\ex{0}$ is injective. \smallskip

\bul \textit{(ii - Compactness of the differential).} Its differential, evaluated at some $v \in \spn{k}$, is given by
\begin{align*}
	& \bpa{\d_v \ro\ex{k}} u = -2N \hspace{-0.1cm}\int_{\R^{d(N-1)}}\hspace{-0.9cm} \d x_2 \cdots \d x_N \p\ex{k}(v) \bpa{\hn(v)-\exc{k}(v)}\iv_{\perp} \vv{u} \p\ex{k}(v).
\end{align*}
	For all $v \in \spn{k}$, $\d_v \ro\ex{k}$ is compact from $\spp$ to $W^{1,1}$, not surjective, and moreover,
\begin{align*}
	\nor{\bpa{\d_v \ro\ex{k}}u}{W^{1,1}}^2 \le c_{v} \nor{u}{\spp} \int \ab{u} \ro\ex{k}(v).
\end{align*}
	When $p > \max(2d/3,2)$, $\d_v \ro\ex{0} u = 0$ implies that $u$ is constant. \smallskip

	\bul \textit{(iii - Local weak-strong continuity).} With the same notations as in Theorem \ref{proposs}, we have $\sqrt{\ro\ex{0}(v_n)} \ra \sqrt{\ro\ex{0}(v)}$ strongly in $H^1(\R^d)$ in addition, where we assumed that $v,v_n \in \spn{0}$.
\end{theorem}
	For a non degenerate $v$, an expression of the associated quadratic form is
 \begin{align}\label{qff}
	 \ps{ u, \bpa{\d_v \ro\ex{k} } u} &  = - 2 \nord{ \bpa{\hn(v)-\exc{k}(v)}^{-\ud}_{\perp,+}  \vv{u} \p\ex{k}(v) }^2 \\
	 & \bhs \bhs + 2 \nord{ \bpa{\hn(v)-\exc{k}(v)}^{-\ud}_{\perp,-}  \vv{u} \p\ex{k}(v) }^2, \nonumber
 \end{align}
where $A_{\pm}$ denote the positive/negative parts of the self-adjoint operator $A$. The spectrum of $\d_v \ro\ex{k}$ is important to study, for instance if there is $u \in \Ker \d_v \ro\ex{k}$, $\ro\ex{k}$ will not change as we move on the direction $u$, and hence it gives the beginning of a branch of potentials having the same density. This was used in \cite{GauBur04} to find numerical counterexamples to the Hohenberg-Kohn theorem for excited states. If $v$ is degenerate, the same problem reduces to the research of directions $u$ such that $\delta_v \ro\ex{k} (u) = 0$.

\begin{proof}[Proof of Theorem \ref{maa}]\tx{ }

	$(i)$ The decomposition $\ro\ex{k} = \rop \rond \p\ex{k}$ is smooth because $\rop$ and $\p\ex{k}$ are so. 
	
	$(ii)$ \bul When $k=0$, the differential is injective because if $\pa{\d_v \ro\ex{0}} u = 0$, then using \eqref{qff} we have
 \begin{align*}
	 \bpa{\hn(v)-\exc{0}(v)}^{-\ud}_{\perp}\vv{u} \p\ex{0}(v) = 0
 \end{align*}
so applying $\bpa{\hn(v)-\exc{0}(v)}^{\ud}_{\perp}$ yields
 \begin{align*}
	P^{\perp}_{\p\ex{0}(v)}  \vv{u} \p\ex{0}(v) = 0.
 \end{align*}
Hence $\vv{u} \p\ex{0}(v) = \alpha \p\ex{0}(v)$ for some constant $\alpha \in \R$, and finally by unique continuation \cite{Garrigue19}, we deduce that $\ssum u_i = \alpha$ and then $u$ is constant.

	\bul The operator $\d_v \ro\ex{k}$ cannot be simultaneously compact, surjective and continuous, by the open mapping theorem. The formula for the differential follows from \eqref{peq} and Lemma \ref{smoothrop}. The bounds follow from Theorem \ref{proposs}, and by the smoothness of $\rop$ implying that $\nor{\d_{\p} \rop}{H^\ell \ra W^{\ell,1}}$ is bounded for any $\ell \in \N$. 

	We remark that for any $u \in L^p+L^{\ii}$ that is not constant, $\ps{u, \bpa{\d_v \ro\ex{0}} u} \sle 0$, hence $\d_v \ro\ex{0} \sle 0$ in the sense of forms. Another way of seeing it is by considering the inequality $\int (v-u)(\ro\ex{0}(v)-\ro\ex{0}(u)) \sle 0$ for any potentials $v,u$ such that $v-u \neq 0$, presented in \cite[Section 2.3]{Garrigue19b}. This also implies $\ps{u, \bpa{\d_v \ro\ex{0}} u} \sle 0$ for any non constant potential.
\end{proof}

The fact that $\im \d_v \ro\ex{k}$ is probably dense in $W^{1,1}$ could suggest to prove a local surjectivity result using \cite[Theorem 2.5.9]{AbrMarRat12} or \cite{RayWal82}. Unfortunately, the compactness of $\d_v \ro\ex{k}$ prevents us from doing so.

\begin{proof}[Proof of Corollary \ref{smallcor}] In this particular case, $\spn{0} = \spp = L^p$. By Theorem \ref{proposs} and Theorem \ref{maa}, $\ro\ex{0}$ is weak-strong continuous. We conclude by applying Lemma \ref{comps} $(iii)$.
 \end{proof}

\subsection{The potential-to-ground energy map}
Finally, the regularity of $v \mapsto \p\ex{k}(v)$ carries to $v \mapsto \exc{k}(v)$.

\begin{proof}[Proof of Corollary \ref{coco}]\tx{ }

	\bul The energy is weakly upper-semicontinuous by the same proof as for the weak lower-semicontinuity of the Lieb functional \cite[Theorem 3.6]{Lieb83b}. It is Lipschitz continuous and concave by \cite[Theorem 3.1]{Lieb83b}.

	\bul We can decompose $\exc{k}(v) = \ps{\p\ex{k}(v), \hn(0) \p\ex{k}(v)} + \int v \ro\ex{k}(v)$, where $v \mapsto \int v \ro\ex{k}(v)$ is $\cC^{\ii}$ because $v \mapsto \ro\ex{k}(v)$ is so, and $(\vp,\phi) \mapsto \ps{\vp,\hn(0)\phi}$ is bilinear so $v \mapsto \exc{0}(v)$ is $\cC^{\ii}$.

	\bul As for the differential, we start by following similar arguments as in \cite[Theorem II.16]{LieSim77b}. By the second form of \eqref{mlm}, for any $v,u \in \spn{k}$ we have $\exc{k}(v) \le \cE_v\pa{\p\ex{k}(u)}$, hence
 \begin{align*}
	 & \bhs \exc{k}(v+u) - \exc{k}(v) \le \cE_{v+u}\bpa{ \p\ex{k}(v)} - \cE_v\bpa{ \p\ex{k}(v) } = \int u \ro\ex{k}(v), \\
	 & \int u  \ro\ex{k}(v+u) = \cE_{v+u}\bpa{ \p\ex{k}(v+u)} - \cE_v\bpa{ \p\ex{k}(v+u) } \\
	 & \bhs\bhs\bhs\bhs\bhs\bhs\bhs\bhs \le \exc{k}(v+u) - \exc{k}(v),
 \end{align*}
hence
 \begin{align*}
	 \int u\bpa{ \ro\ex{k}(v+u)-\ro\ex{k}(v)} \le \exc{k}(v+u) - \exc{k}(v) - \int u \ro\ex{k}(v) \le 0.
 \end{align*}
By the Gagliardo-Nirenberg inequality, if $q \in \seg{1,d/(d-1)}$ (with $d/(d-1) \df +\ii$ if $d=1$), then for any $f \in W^{1,1}(\R^d)$,
 \begin{align*}
 \nor{f}{L^{q}} \le c \nor{\na f}{L^1}^{d\pa{1-\inv{q}}} \nor{f}{L^1}^{1-d\pa{1-\inv{q}}}.
 \end{align*}
	Take $q \df p/(p-1) \in \seg{1,d/(d-1)}$. Since $\ro\ex{k}$ is $\cC^{\ii}$, we have
 \begin{align*}
	 & \ab{ \exc{k}(v+u) - \exc{k}(v) - \int u \ro\ex{k}(v) } \le  \ab{ \int u\bpa{ \ro\ex{k}(v+u)-\ro\ex{k}(v)}} \\
	 & \bhs \le c \nor{u}{\spp} \nor{\ro\ex{k}(v+u)-\ro\ex{k}(v)}{L^1 \cap L^{q}} \\
	 & \bhs\le c \nor{u}{\spp} \pa{\nor{\ro\ex{k}(v+u)-\ro\ex{k}(v)}{L^1} + \nor{\ro\ex{k}(v+u)-\ro\ex{k}(v)}{L^{q}}} \\
 & \bhs\le c \nor{u}{\spp}^{1+ \min\pa{1, d\pa{1-\inv{q}}}},
 \end{align*}
and $q > 1$ so $1+ \min\pa{1, d\pa{1-\inv{q}}} > 1$ and this proves the existence of the differential. 

	\bul We show that $\exc{0}(v)$ is strictly increasing on $\spf{0}$. Take $u \in \spf{0}$, $v \in \spa$ with $v \le u$, and $v \sle u$ on a set of positive measure. Take a ground state $\p_u$ of $\hn(u)$. By unique continuation \cite[Remark 1.6]{Garrigue19}, the nodal set of $\ro_{\p_u}$ has zero volume, hence $\ab{ \acs{v \ro_{\p_u} \sle u \ro_{\p_u}}} >0$ and 
 \begin{align*}
	 \exc{0}(v) & \le \cE_0\bpa{\p_u} + \int v \ro_{\p_u} \sle \cE_0\bpa{\p_u} + \int u \ro_{\p_u} = \exc{0}(u).  
 \end{align*}

\bul Eventually, we prove by contradiction that $\exc{0}$ is strictly concave on $\spf{0}$. Let $v,u \in \spf{0}$, we start from the point $u$ and look at the (half line) direction $v-u$. By using the concavity of $\exc{0}$ and formula \eqref{degdeg}, we have 
\begin{align*}
	\exc{0}(v) - \exc{0}(u) & \le {^+}\delta_u \exc{0} (v-u) = \myinf{\p \in \Ker \pa{\hn(u) - \exc{0}(u)} \\ \int \ab{\p}^2 = 1} \int \ro_{\p} (v-u).
\end{align*}
The minimizing set in the right hand side of the previous inequality is compact, let us denote by $\p_{u,v}$ one of the minimizers. This yields
\begin{align*}
\exc{0}(v) - \exc{0}(u) \le \cE_v\bpa{\p_{u,v}} - \cE_u\bpa{\p_{u,v}} = \cE_v\bpa{\p_{u,v}} - \exc{0}(u).
\end{align*}
Let us assume that we have equality above, then $\exc{0}(v) = \cE_v(\p_{u,v})$. The following is the same argument as the second part of the Hohenberg-Kohn theorem \cite{HohKoh64}, as presented in \cite[Proof of Theorem 2.1]{Garrigue19b} for instance. We know that $\p_{u,v}$ is a ground state for $\hn(v)$, hence it satisfies its Schr\"odinger's equation $\hn(v) \p_{u,v} = \exc{0}(v) \p_{u,v}$. Substracting with $\p_{u,v}$'s own Schr\"odinger's equation, we obtain $\pa{\exc{0}(u) - \exc{0}(v) +  \ssum (v-u)_i} \p_{u,v} = 0$, and by strong unique continuation \cite{Garrigue18,Garrigue19}, that $v = u + \bpa{\exc{0}(v) - \exc{0}(u)}/N$.
\end{proof}

 \section{Proofs of Theorem \ref{helf} and Corollary~\ref{propcandif}}

 \begin{proof}[Proof of Theorem \ref{helf}]\tx{ }
	 
	 \bul We consider the map $\alpha \mapsto \hn(v+ \alpha u)$. Our starting point is \cite[Theorem 1.4.4, Corollary 1.4.5]{Simon15}, stating that close to $\alpha = 0$, the singularities arising from degeneracies are removable. In the given reference this is stated for $\cH = L^2(\R^d)$ but this applies for any separable Hilbert space. It justifies the existence of $\dim \cD\ex{k}(v)$ maps $E_i : \R \ra \R$, $\phi_i : \R \ra \cH$, $i \in \acs{1, \dots \dim \cD\ex{k}(v)}$, analytic in a neighborhood of $0$, $E_i$ being the eigenvalues of $\hn(v)$ such that $E_i(0) = \exc{k}(v)$, and their associated orthonormal eigenfunctions $\phi_i$. By analyticity, $E_i$ and $\phi_i$ can be expressed in the so-called Rayleigh-Schr\"odinger series, which coefficients are the derivatives $(n!)\iv \pa{{^+}\d^n E_i / \d \alpha^n}(0)$ and $(n!)\iv \pa{{^+}\d^n \phi_i / \d \alpha^n} (0)$. The functionals are hence infinitely Dini differentiable.

	 \bul From degenerate perturbation theory, one can deduce the formulas \eqref{degdeg}, the $\dim \cD\ex{k}(v) \times\dim \cD\ex{k}(v)$ matrix having the information of the first order is
 \begin{align*}
	 P_{\p\ex{k}(v)} \vv{u} P_{\p\ex{k}(v)},
 \end{align*}
	 see for instance \cite{Hir69}. In particular, the right eigenbasis $\phi_i(0)$ making $\alpha \mapsto \phi_i(\alpha)$ analytic is given by an eigenbasis of the above matrix. The eigenvalues give the first derivatives of the energy. 

	 A priori we had to work in the complex sphere of normalized complex eigenstates, but since the potentials are real and there is no magnetic field, we can choose real eigenstates, and the optimization over complex $\lambda$'s can absorb the complex factors of the first optimization set. However, we cannot further simplify by optimizating over real $\lambda$'s, because this would decrease the optimization set.

\bul The map $u \mapsto {^+}\delta_v \exc{0}(u)$ is concave because this is a minimum over a familly of linear functionals, it is weakly upper semi-continuous by Mazur's theorem and because the sets $\acs{u \st ^+\delta_v \exc{0} (u) \ge \alpha}$ are closed. See for instance \cite[Theorem 3.6]{Lieb83b} for more details on this argument.

\bul We prove \eqref{coucou}. In this case we minimize over the Bloch sphere projective space $P \vect\pa{\p_1,\p_2}$, we have
\begin{align*}
	{^+}\delta_v \exc{0}(u) = \mymin{a, b \in \C \\ \int \ab{a \p_1 + b \p_2}^2 = 1 } \int u \ro_{a \p_1 + b \p_2},
\end{align*}
	 and the constraint on $a,b$ reduces to $\ab{a}^2 + \ab{b}^2 = 1$. We can take the parametrization $a = \pa{\cos t} e^{i \eta}$, $b = \pa{ \sin t} e^{i(\eta + \theta)}$ and $\p_1,\p_2$ are real. We define $A \df \ud \int u \pa{\ro_{\p_1}-\ro_{\p_2}}, B \df \ps{\p_1,\vv{u}\p_2}$, and have
 \begin{align*}
 & ^+\delta_v \exc{0}(u) \\
 & \hs\hs\hs = \mymin{t,\theta \in [0,2 \pi] } (\cos t)^2 \int u \ro_{\p_1} + (\sin t)^2 \int u \ro_{\p_2} +  \ps{\p_1,\vv{u}\p_2} \cos \theta \sin(2t) \\
 &\hs\hs\hs  = \ud  \int u \pa{\ro_{\p_1} + \ro_{\p_2}} + \mymin{t,\theta \in [0,2 \pi] } A \cos t  + B \cos \theta \sin t  .
 \end{align*}
 Optimizing over $t$ yields the optimal value $t^* \in \pi \N + \arctan\pa{B (\cos \theta)/A}$ and using the classical formula for $\cos \arctan$ and $\sin \arctan$, we get
 \begin{align*}
A \cos t^*  + B \cos \theta \sin t^* =\pm \f{A^2 + B^2 (\cos \theta)^2}{A \sqrt{1+ (B (\cos \theta)/A)^2}} = \pm \sqrt{ A^2 + (\cos \theta)^2 B^2}.
 \end{align*}
Finally optimizing over $\theta$ gives \eqref{coucou}. We could also have computed the eigenvalues from
 \begin{align*}
P_{\p\ex{k}(v)} \vv{u} P_{\p\ex{k}(v)}	 = \mat{ \int u \ro_{\p_1} & \int \p_1 \p_2 \vv{u} \\ \int \p_1 \p_2 \vv{u} & \int u \ro_{\p_2}},
 \end{align*}
but this would not have given us the rotation enabling to compute the eigenvectors from the initial vectors.
  \end{proof}

\begin{remark}
For higher derivatives, we can write such variational formulas. For instance by defining the resolvent $K \df (\hn(v)-\exc{0}(v))\iv_{\perp}$, we have
 \begin{align*}
	 {^+}\delta^2_v \exc{0} (u) & = -\myinf{\p\tx{ minimizes }{^+}\delta_v \exc{0}(u)}  \ps{\p, \vv{u} K  \vv{u} \p}, \\
	 {^+}\delta^3_v \exc{0} (u) & = \myinf{\p\tx{ minimizes }{^+}\delta^2_v \exc{0}(u)}  \ps{\p, \vv{u} K  \vv{u} K\vv{u} \p}.
 \end{align*}
\end{remark}

\begin{proof}[Proof of Corollary \ref{propcandif}]\tx{ }
Our assumptions on $k$ enable to write $\exc{k}(v)$ as a minimum or as a maximum, but without involving min-max formula. In this proof we assume that the formula is given by a minimum, as is the case when $k=0$ or when $\exc{k-1}(v) \sle \exc{k}(v)$, but the case $\exc{k}(v) \sle \exc{k+1}(v)$ is similar.
	
\bul $i)$ Take $\lambda \in \R$, we have
 \begin{align*}
	 & {^+}\delta_v \exc{k}(\lambda u) - \lambda \pa{{^+}\delta_v \exc{k}}(u) \\
	 & \bhs = \pm \lambda \times \left\{
\begin{array}{llr}
	0 & \tx{if } \lambda > 0, \\
	\mysup{\p \in \cD\ex{k}(v)} \int u \ro_{\p} - \myinf{\p \in \cD\ex{k}(v)} \int u \ro_{\p} & \tx{if } \lambda \sle 0.
\end{array}
\right.
 \end{align*}
Hence ${^+}\delta_v \exc{k}$ is linear in the direction $u$ if and only if $\int u \ro_{\p} =: c$ is constant in $\p \in \cD\ex{k}(v)$. In this case, for $\p, \Phi \in \cD\ex{k}(v)$, we take $a,b \in \C$ such that $1 = \int \ab{a \p + b \Phi}^2 = \ab{a}^2 + \ab{b}^2 + 2 \pa{\re a \ov{b}} \ps{\p,\Phi}$ and such that $\re a \ov{b} \neq 0$, we compute
 \begin{align*}
 c =  \int u \ro_{a \p + b \Phi} = c + 2 \pa{\re a \ov{b}} \pa{N \int_{\R^{d(N-1)}} u(x_1) \p \Phi - \ps{\p,\Phi} c},
 \end{align*}
hence the last equivalence.

\bul $ii)$ Assume that ${^+}\delta_v \exc{k}$ is linear in all directions. Thus $\int u \ro_{\p}$ is constant in $\p$ and in $u$, this implies that $\ro_{\p}$ is constant in $\p$. But then ${^+}\delta_v \exc{k}(u) = \int u \ro$ and we remark that it is linear, hence differentiable. In this case and from $i)$, we have
 \begin{align*}
	\int_{\R^{d}} u \pa{ \ps{\p,\Phi} \ro_{\p} - N \int_{\R^{d(N-1)}} \p \Phi} = 0
 \end{align*}
	uniformly in $u$, and we can conclude.

	\bul $iii)$ Let us denote by $(\vp_i)_{1 \le i \le K}$ ground and excited states of $-\Delta + v$. Let us treat $k=0$ first. Since $D \df \dim \cD\ex{k}(v) \ge 2$, then the Fermi level of $-\Delta + v$ is degenerate, and $\vp_{D}$ and $\vp_{D+1}$ both belong to it. We consider $\p \df \vp_{D} \wedge_{i=1}^{D-1} \vp_i \in \cD\ex{0}(v)$ and $\Phi \df \vp_{D+1} \wedge_{i=1}^{D-1} \vp_i \in \cD\ex{0}(v)$. We assume that ${^+}\delta_v \exc{k}$ is differentiable, hence the degeneracy is broken in no direction at first order. By applying $ii)$, we deduce that $\vp_{\ell}\vp_{m}=0$, contradicting $\ab{\acs{\vp_{\ell}=0}}=\ab{\acs{\vp_{m}=0}}=0$ implied by unique continuation \cite{JerKen85}.

Let us now treat the case $k \ge 1$. In the case $N=1$, the concerned level is always degenerate. When $N \ge 2$, see \cite[Section 4.1, Figure 1]{Garrigue21} for more precisions on the possible configurations. In any case, we can use a similar construction as we did for $k=0$, where $\p$ and $\Phi$ are taken in the $N$-body Fermi level and have only one difference in the orbitals filling. Indeed, since $k$ is in the first excited eigenspace of the $N$-body operator, either the Fermi level of the one-body operator, or the next one, is degenerate.
 \end{proof}

\section*{Appendix A: basic inequalities on potentials}
We recall here in Lemma \ref{continj} several well-known facts about potentials. 

\begin{lemma}\label{continj} Take $v,w \in (L^p+L^{\ii})(\R^d)$.

	$(i)$ Taking $p$ as in \eqref{expot}, we have
\begin{align}\label{ude}
	 & \norop{\bpa{-\Delta_{\R^{dN}}+1}^{-\ud} \vv{v} \bpa{-\Delta_{\R^{dN}}+1}^{-\ud}}  \nonumber \\
		& \bhs \bhs \bhs \leq N \norop{\bpa{-\Delta_{\R^{d}}+1}^{-\ud} v \bpa{-\Delta_{\R^{d}}+1}^{-\ud}}\nonumber  \\
		& \bhs \bhs \bhs \leq c_{d,p} N \nor{v}{L^p+L^{\ii}}.
 \end{align}

$(ii)$ Let $\cC \subset \C$, be a contour in the complex plane which is such that $\dist\pa{z,\sigma(\hn(v))} \ge \eta > 0$ uniformly in $z \in \cC$. 
Let $p$ be as in \eqref{expot}, then the operators
 \begin{align*}
	 \del^{-\ud} \bpa{\hn-z}  \del^{-\ud}, \hs\hs\hs\hs\hs\hs\hs\hs\hs \del^{\ud} \bpa{\hn-z}\iv  \del^{\ud}
 \end{align*}
are uniformly bounded in $z \in \cC$.

$(iii)$ Let $v \in \spn{0}$, and $p$ as in \eqref{expot}. For $u \in L^p+L^{\ii}$ such that $\nor{u}{L^p+L^{\ii}}$ is small enough, we have $v+u \in \spn{0}$.
\end{lemma}
\begin{proof}
	\tx{ }
	
	$(i)$ \bul If $p$ is as in \eqref{expot}, we have
 \begin{align*}
	 & \norop{(-\Delta_{\R^{dN}} +1)^{-\ud}\vv{v} (-\Delta_{\R^{dN}} +1)^{-\ud}} \\
	 & \bhs\bhs \le \sum_{i = 1}^N \norop{(-\Delta_{\R^{dN}} +1)^{-\ud}v_i (-\Delta_{\R^{dN}} +1)^{-\ud}} \\
	 & \bhs\bhs \le N\norop{(-\Delta_{\R^{d}} +1)^{-\ud}v (-\Delta_{\R^{d}} +1)^{-\ud}},
 \end{align*}
	where we used that $\norop{(-\Delta_i)^{\ud} (-\Delta_{\R^{dN}} +1)^{-\ud}} =1$.

	\bul Let us write $v = v_p + v_{\ii}$. We have
\begin{align}\label{klk}
 & \norop{\del^{-\ud}v\del^{-\ud}} \nonumber \\
 & \bhs\le \norop{\del^{-\ud}v_p\del^{-\ud}} \nonumber\\
	& \bhs \bhs\bhs\bhs\bhs + \norop{\del^{-\ud}v_{\ii}\del^{-\ud}}\nonumber \\
 & \bhs\le \norop{\del^{-\ud}v_p\del^{-\ud}} + \nor{v_{\ii}}{L^{\ii}} \nonumber\\
 & \bhs\le \norop{\sqrt{\ab{v_p}} \del^{-\ud}}^2 + \nor{v_{\ii}}{L^{\ii}}.
\end{align}
 In the last inequality, we used that
 \begin{align*}
	 \del^{-\ud} v_p \del^{-\ud} = \del^{-\ud} \sqrt{\ab{v_p}} \sgn(v_p) \sqrt{\ab{v_p}}\del^{-\ud},
 \end{align*}
	where $\sgn(v)$ is equal to $1$ if $v > 0$, $-1$ if $v \sle 0$ and $0$ if $v =0$, it satisfies $\norop{\sgn(v)} \le 1$, hence
 \begin{align*}
	 \norop{\del^{-\ud} v_p \del^{-\ud}} \le \norop{\sqrt{\ab{v_p}} \del^{-\ud}}^2.
 \end{align*}
	As for the first term in \eqref{klk}, for $d \ge 3$, with $p = d/2$ we have
 \begin{align*}
	 \norop{\sqrt{\ab{v_p}} \del^{-\ud}} &  \le \norop{\sqrt{\ab{v_p}} (-\Delta)^{-\ud}}  \le c_{d} \nor{\sqrt{\ab{v_p}}}{L^{2p}}\\
	 & = c_{d} \sqrt{\nor{v_p}{L^{p}}},
 \end{align*}
where we used the Hardy-Littlewood-Sobolev inequality \cite[Theorem 4.3]{LieLos01} in the last inequality. For $d \in \acs{1,2}$, we can use the Kato-Seiler-Simon inequality \cite[Theorem 4.1]{Simon79} to get 
 \begin{align*}
	 \norop{\sqrt{\ab{v_p}} \del^{-\ud}} & \le \nor{\sqrt{\ab{v_p}} \del^{-\ud}}{\sch_{2p}} \\
	 & \le (2\pi)^{-d/(2p)} \nor{\sqrt{\ab{v_p}}}{L^{2p}} \nor{\bpa{\ab{x}^2+1}^{-\ud}}{L^{2p}} \\
	 & \le c_{d,p} \sqrt{\nor{v_p}{L^{p}}}.
 \end{align*}

	$(ii)$ 
	 Take $c \ge 0$ and let us define $A \df \ssum v_i + \ssumd w_{ij}$. We remark that
 \begin{align*}
	 H + c = (-\Delta+c)^{\ud} \pa{ 1+ (-\Delta+c)^{-\ud} A (-\Delta+c)^{-\ud}}(-\Delta+c)^{\ud},
 \end{align*}
	hence we only need to show that $\norop{(-\Delta+c)^{-\ud} A (-\Delta+c)^{-\ud}} \sle 1$. For instance we will show that
 \begin{align*}
	 \norop{(-\Delta_{\R^d} + c)^{-\ud} v (-\Delta_{\R^d} + c)^{-\ud}} \le \norop{\sqrt{\ab{v}} (-\Delta_{\R^d} + c)^{-\ud}}^2
 \end{align*}
	is as small as we want. For any $\ep > 0$, there exists $c_{\ep} \ge 0$ such that $\ab{v} \le \ep (-\Delta) + c_{\ep}$ in the sense of forms, hence for all $u \in \cC^{\ii}$, we have
 \begin{align*}
	 \nord{\sqrt{\ab{v}} (-\Delta + c)^{-\ud} u}^2 & \le  \ep \nord{(-\Delta)^{\ud} (-\Delta + c)^{-\ud} u}^2 + c_{\ep} \nord{(-\Delta + c)^{-\ud} u}^2 \\
	 &  \le \pa{ \ep + \f{c_{\ep}}{c}} \nord{u}^2.
 \end{align*}
 We can first choose $\ep$ small and then choose $c$ large so that the quantity $\norop{\sqrt{\ab{v}} (-\Delta + c)^{-\ud}}$ is arbitrarily small. \\

$(iii)$ The statement follows from the resolvent formula
\begin{align*}
 & \pa{z-\hn(v+u)}\iv - \pa{z-\hn(v)}\iv \\
& \bhs= \pa{z-\hn(v+u)}\iv \vv{u} \pa{z-\hn(v)}\iv,
\end{align*}
and Cauchy's formula
 \begin{align*}
 \indic_{\acs{\exc{0}(v)}} = \inv{2i\pi} \oint_{\cC} \f{\d z}{z-\hn(v)},
 \end{align*}
	see for instance \cite{ReeSim4,Lewin18}.
\end{proof}

\section*{Appendix B: weak-strong continuity and compactness}

We recall here relations between weak-strong continuity and compactness. Following \cite[Definition 7.6]{Hohage}, we say that a map is compact if it maps bounded sets into relatively compact sets. The link between ill-posedness of a problem and its linearization can be involved, see for instance \cite{Schock02} and \cite[Appendix]{EngKunNeu89}. We start by considering standard results, and adapt them to the case when the image space is an embedded submanifold.

\begin{lemma}\label{comps}
Let $X$ and $Y$ be Banach spaces, $U \subset X$ an open set, $M \xhookrightarrow{}Y$ a closed embedded submanifold of $Y$, and a map $f : U \ra M$.
\begin{enumerate}[label=(\roman*),leftmargin=*]
\item\label{frs} If $f$ is compact, continuous and differentiable on $U$, then $\d_x f$ is compact for any $x \in U$.

\item\label{tro} If $U = X$ is the dual of a Banach space, and if $f$ is weak-strong continuous, then $f$ is compact.

\item\label{scd} If $f$ is compact and $M$ is infinite-dimensional, then $f(X)$ is a countable union of compact sets, and $f(X)$ has empty interior.

\item\label{qa} If $f$ is compact and $X$ is infinite-dimensional, then $f\iv$ is discontinuous.
\end{enumerate}	
\end{lemma}

 \begin{proof}
The only difference in the proof, compared to the case $M=Y$, is $(i)$.

$\ref{frs}$ In the case $M = Y$, this is proved in \cite{Hohage}. We apply it to $\iota_{M \ra Y} \rond f : U \ra Y$ and get that $\d_x \pa{\iota_{M \ra Y} \rond f}\pa{ X \cap \acs{ \norop{\cdot} \le 1}} = \iota_{\Td_{f(x)} M \ra Y} \rond \pa{\d_x f} \pa{X \cap \acs{ \norop{\cdot} \le 1} }$ is compact. A map is proper if preimages of relatively compact open sets are relatively compact open sets \cite[Definition 16.26]{Souriau12}. One can prove that for a Banach space $F$ and a closed subset $E \subset F$, the inclusion map $E \xhookrightarrow{} F$ is proper. Since $\iota_{\Td_{f(x)} M \ra Y}$ is proper, then $\pa{\d_x f} \pa{X \cap \acs{ \norop{\cdot} \le 1} }$ is relatively compact. We remark that we only used that $M \xhookrightarrow{}Y$ is an embedded submanifold of $Y$, we did not use the closed condition.

$\ref{tro}$ Let $G \subset B_0(r) \subset X$ be a bounded set and $x_n \in G$ a sequence. By Banach-Alaoglu's theorem, $x_n \wra x$ for some $x \in B_0(r)$ and up to a subsequence. By weak-strong continuity of $f$, $f(x_n) \ra f(x)$ strongly.

$\ref{scd}$ 
 We define the sets $X_r \df X \cap \acs{ x \in X \st \nor{x}{X} \le r}$, for $r \ge 0$. Since $f$ is compact, then the $\ov{f(X_r)}$'s are compact and thus have empty interiors by Riesz's theorem \cite[Theorem 6.5]{Brezis10}, which applies in our case because $M$ is locally a normed vector space. We have 
\begin{align*}
 f(X) = \cup_{r\in \N} f(X_r) \subset \cup_{r\in \N} \ov{f(X_r)}.
\end{align*}
 Finally, by Baire's theorem \cite[Theorem 2.1]{Brezis10} $f(X)$ has empty interior. We recall that a closed subset of a compact space is compact.

$\ref{qa}$ Let $B \subset X$ be a ball, $f(B)$ is relatively compact. Assuming that $f\iv$ is continuous, $f\iv\pa{f(B)} \supset B$, is also relatively compact, and hence $B$ as well. 
 But this contradicts \cite[Theorem 6.5]{Brezis10}. The inverse $f\iv$ is thus discontinuous.
 \end{proof}

Here is a summary of the relations between compactness and weak-strong continuity for a map and its differential.
\begin{center}
\begin{tikzcd}
	f \tx{ compact} \arrow[r,Rightarrow] \arrow[swap,d,Leftarrow,"\tx{if $U=X$ is a dual }"] & \d_x f \tx{ compact } \forall x \arrow[d,Rightarrow] \arrow[d,Leftarrow,"\tx{if $X$ reflexive}",xshift=2ex] \\
	f \tx{ locally weak-strong } \cC 
	& \d_x f \tx{ weak-strong $\cC$ } \forall x
\end{tikzcd}
\end{center}
We also remark that $\d_x f$ weak-strong continuous for any $x \in U$ does not imply that $f$ is weak-strong continuous, a simple counterexample is $L^2(\R^n) \ni x  \mapsto \nord{x}^2$, and this is also the case for $v \mapsto \p(v)$. 

\bibliographystyle{siam}
\bibliography{biblio}
\end{document}